\documentclass[lettersize,journal]{IEEEtran}
\ifCLASSINFOpdf
\else
\fi
\usepackage{stfloats}
\usepackage{cite,amssymb}
\usepackage{mathrsfs}
\usepackage[pdftex]{graphicx}
\usepackage{algorithm}
\usepackage[cmex10]{amsmath}
\usepackage[T1]{fontenc}
\usepackage{textcomp}
\usepackage{algpseudocode}
\usepackage{array,balance,bm}
\usepackage{nicematrix}
\usepackage{arydshln} 
\usepackage{lipsum}
\usepackage{booktabs}
\usepackage{siunitx} 

\usepackage{amsfonts}
\usepackage{exscale}
\usepackage{relsize}
\usepackage{multicol,color}

\usepackage{amsthm}

\usepackage[caption=false]{subfig}

\newtheorem{theorem}{Theorem}
\newtheorem{lemma}{Lemma}
\newtheorem{corollary}{Corollary}
\newtheorem{proposition}{Proposition}

\theoremstyle{remark}
\newtheorem{remark}{Remark}

\newcommand{\be}{\begin{equation}}
	\newcommand{\ee}{\end{equation}}

\allowdisplaybreaks[4]

\begin{document}

	\title{Coverage Analysis of Rydberg Atom Quantum Receiver Arrays: A Stochastic Geometry Approach}
	\author{Dongnan Xia, Cunhua Pan, Senior Member, IEEE, Hong Ren, Member, IEEE, Dongsheng Sui, Qihao Peng, Member, IEEE, and Jiangzhou Wang, Fellow, IEEE%
	\thanks{Dongnan Xia, Cunhua Pan, Hong Ren, Dongsheng Sui, and Jiangzhou Wang are with the National Mobile Communications Research Laboratory, Southeast University, Nanjing, China (e-mail: \{dnxia, cpan, hren, 230258960, j.z.wang\}@seu.edu.cn).}%
	\thanks{Qihao Peng is with the School of Computer Science and Electronic Engineering, University of Surrey, Guildford, U.K. (e-mail: q.peng@surrey.ac.uk).}%
	\thanks{Corresponding authors: Cunhua Pan and Hong Ren.}}
	\maketitle
	\bstctlcite{setting}
	
	
	\begin{abstract}
		Rydberg atomic quantum receivers (RAQRs) offer quantum-limited sensitivity and broadband tunability. 
		It is not obvious whether this device-level advantage also improves network reliability, since in dense deployments, aggregate interference can push the atomic transducer out of its small-signal regime. This paper addresses the question by embedding the RAQR front end into a stochastic geometry (SG) coverage analysis. Starting with the atomic master equation and balanced coherent optical detection, we derive a third-order complex baseband model that retains both the linear gain and the leading  cubic nonlinearity. A Bussgang decomposition converts the per-element nonlinear response into an equivalent linear gain plus a distance-dependent distortion noise. Using this equivalent model, we derive the post maximal-ratio combining (MRC) SINR and obtain tractable expressions for the conditional and spatially averaged coverage probabilities. 
		The analytical results show that RAQRs outperform conventional receivers in sparse deployments. However, when the base station (BS) density becomes large, nonlinear distortion reduces this advantage and may make RAQRs perform worse.
		Simulation results validate the analytical expressions and confirm that the central design tradeoff is between linear gain and cubic nonlinearity.
	\end{abstract}
	
	
	\section{Introduction}\label{sec:introduction}
	The continuing evolution of wireless communication and sensing systems is placing growing pressure on the receiver front-end, particularly in scenarios that demand weak-signal detection, high sensitivity, broad tunability, and compact hardware.
	Conventional electronic receivers, however, are still constrained by the antenna size, the complexity of the front-end circuitry, and thermal noise introduced along the analog chain, which together limit their sensitivity~\cite{rappaport2010wireless}.
	Rydberg atomic quantum receivers (RAQRs) provide a different receiving mechanism: they use the interaction between incident radio-frequency (RF) fields and highly excited atomic states, and read out the atomic response optically rather than through a conventional electronic front-end~\cite{fancher2021rydberg, sedlacek2012microwave}.
	Experimental demonstrations have reported sensitivities on the order of tens of $\mathrm{nV/cm/\sqrt{Hz}}$ and minimum detectable fields at the sub-$\mathrm{nV/cm}$ level~\cite{jing2020atomic}.
	These measurements reflect several practical strengths of RAQRs: self-calibration based on atomic constants, broad spectral coverage, and operation without a metallic antenna front end.
	These features support their use in quantum sensing and wireless reception~\cite{holloway2014broadband,
	anderson2021self,gong2025raqr,cui2024rare}.
	At the same time, practical sensitivity is affected by quantum projection noise, photodetector noise, and laser intensity and phase noise~\cite{wang2023noise, tang2025noise}, so device-level sensitivity alone does not determine network reliability.

	
	Recent studies have expanded the study of Rydberg receivers from initial experiments to wireless receiver architectures and corresponding system models.
	Early studies showed that atom based receivers could recover digitally
	modulated signals~\cite{meyer2018digital,song2019digital} and amplitude
	modulated signals in single and multiband settings
	\cite{li2022am,holloway2020multiband}. These results indicated that
	Rydberg receivers can process information carrying microwave signals
	rather than only detect static fields.
	Subsequent demonstrations extended the supported waveforms to phase
	modulation~\cite{cai2023phase} and quadrature amplitude modulation
	\cite{nowosielski2024qam}.
	This progress was strengthened by continuous frequency atomic heterodyne detection, which removed the frequency discreteness of earlier Rydberg heterodyne schemes and made this technique more suitable for practical receiver implementation~\cite{liu2022continuous}.
	More recently, research has extended from isolated reception scenarios to broader communication system design. Atomic multiple input multiple output (MIMO) receiver architectures have been proposed and analyzed~\cite{cui2025atomicmimo}. Electromagnetic modeling, capacity analysis, and multi-band RAQ-MIMO architectures for Rydberg atomic receiver systems have also been reported~\cite{yuan2025mimo,zhu2025general,zhu2025raq}. In addition, dynamic signal models based on transfer functions have been developed to characterize the response to time varying inputs~\cite{zhu2025dynamic}.
	In parallel, Gong \emph{et al.} developed the first equivalent baseband model for superheterodyne RAQRs in the SISO case, which gave a closed form expression for the linear gain and made the model convenient for communication performance analysis~\cite{gong2025raqr}.
	These studies showed that RAQRs were being studied not only as quantum sensors, but also as candidate receivers for wireless communication systems.
	Recent work has further extended RAQR research toward richer communication scenarios. For example, Rydberg atomic quantum receivers have been investigated for enhanced ground--satellite direct access, simultaneous wireless information and power transfer~(SWIPT)-MIMO networks, and RIS-assisted atomic MIMO reception~\cite{peng2026enhanced,peng2025swipt,peng2025ris}. 
	These studies expanded the range of communication related problems considered in RAQR research, but they still did not address reliability in large scale wireless networks.
	Existing studies are still mainly focused on isolated links, deterministic propagation conditions, or weak signal linear operation. These cannot show whether the advantages of RAQRs remain in random interference environments.


	Most existing studies on Rydberg atomic receivers remain at the device or single-link level and are usually developed under deterministic propagation assumptions.
	A network-level evaluation must instead account for the random serving distance and the interference generated by other active users.
	Stochastic geometry is commonly used for this purpose, and provides tractable signal-to-interference-plus-noise ratio~(SINR) and coverage characterizations for cellular networks~\cite{andrews2011tractable, haenggi2012stochastic}, uplink networks with fractional power control~\cite{novlan2013uplink}, heterogeneous deployments~\cite{dhillon2012hetnet}, and integrated sensing and communication networks~\cite{gan2024isac}.
	These results, however, are built around a linear receiver response.

	
	This linear-receiver assumption is not directly compatible with RAQR hardware.
	In a Poisson network, the input at each RAQR element of a typical base station (BS) is the sum of contributions from the desired user equipment (UE) and from active interfering UEs in the same network tier.
	As the BS density grows, the typical serving distance shrinks and interferers become closer to the receiver, both of which raise the total per-element input power.
	Atomic superheterodyne receivers have a finite linear dynamic range~\cite{wu2023ldr}, so once the incident field is strong enough, the atomic transducer enters gain compression and its response no longer follows the small‑signal slope.
	The first-order baseband model used in~\cite{gong2025raqr} is accurate in the weak-field regime but does not represent this strong-input behavior, which is the regime that random interference naturally produces in dense deployments.
	The first‑order baseband model from [2] is valid only for weak fields and does not account for the strong input that random interference generates in dense deployments.
	A higher-order baseband model that retains the leading nonlinear term is therefore needed before any network-level coverage analysis of RAQRs becomes meaningful.

	
This paper addresses these two requirements jointly.
We first extend the existing first-order RAQR baseband model to third order, which captures the leading nonlinear term and remains tractable.
We then embed this model into a stochastic-geometry coverage analysis to identify when the low-noise benefit of RAQR arrays is preserved and when nonlinear distortion becomes dominant.
Existing RAQR studies have considered system overviews, atomic MIMO reception, and dynamic signal models~\cite{gong2025raqr, cui2025atomicmimo, zhu2025dynamic}, while stochastic-geometry coverage analysis usually assumes a linear receiver response~\cite{andrews2011tractable, gan2024isac}.
However, existing studies have not analyzed coverage performance when random aggregate input power drives the RAQR into the nonlinear regime.
To the best of our knowledge, this nonlinear reception effect has not been incorporated into tractable RAQR coverage analysis.

The main contributions are summarized as follows:
\begin{itemize}
	\item \textbf{Third-order baseband signal model.}
	Starting with the four-level Lindblad master equation and the balanced coherent optical detection (BCOD) readout, we derive a complex baseband model for each RAQR element.
	The model keeps the linear coefficient $c_1$ and the leading cubic coefficient $c_3$, and therefore reduces to the existing first-order model in the weak-field regime while capturing nonlinear distortion under strong aggregate input power.
	
	\item \textbf{Bussgang linearization and SINR characterization.}
	We apply the Bussgang decomposition to the cubic RAQR response and obtain an equivalent gain $\kappa(r)$ and a distortion variance $\sigma_d^2(r)$.
	This converts the nonlinear receiver output into a tractable linear form with an additional distortion term.
	Based on this representation, we derive the post-MRC SINR in Theorem~\ref{thm:SINR}, where the effective noise $\sigma_{\mathrm{eff}}^2(r)$ includes both photodetection noise and Bussgang distortion.
	
	\item \textbf{Coverage probability and benchmark analysis.}
	We derive the conditional and spatially averaged coverage probabilities for the RAQR uplink model.
	A conventional linear receiver is also analyzed under the same network setting, which gives a direct benchmark for quantifying the RAQR gain and the loss caused by nonlinear distortion.
	We further discuss the noise-limited case and the effect of array coupling, which help explain the regimes where RAQR arrays benefit from low receiver noise and weak inter-element coupling.
	
	\item \textbf{Numerical validation and operating regimes.}
	Monte Carlo simulations are used to validate the analytical coverage expressions for different BS densities, SINR thresholds, and array sizes.
	The results compare different atomic species and array configurations, and show when the low-noise benefit of RAQRs is dominant and when the ratio $|c_3/c_1|$ becomes the limiting factor.
\end{itemize}
	
	\textbf{Notation:} 
	Bold lowercase and uppercase letters denote vectors and matrices, 
	respectively. 
	The operators $(\cdot)^T$, $(\cdot)^H$, and $\|\cdot\|$ denote 
	transpose, Hermitian transpose, and Euclidean norm, respectively. 
	The symbol $\mathbb{E}[\cdot]$ denotes expectation, and 
	$\mathcal{CN}(\mu,\sigma^2)$ denotes a circularly symmetric complex 
	Gaussian random variable with mean $\mu$ and variance $\sigma^2$. 
	$\|\mathbf{x}\|$ denotes the Euclidean distance from $\mathbf{x}$ to the origin. 
	Other symbols are defined when they first appear.
	
	The rest of the paper is organized as follows. 
	Section~\ref{sec:signal_model} develops the third-order baseband 
	model of the RAQR element. 
	Section~\ref{sec:SG_framework} presents the stochastic geometry 
	analysis and derives the coverage probability. 
	Section~\ref{sec:special_cases} discusses several special cases 
	and extensions. 
	Section~\ref{sec:simulation} reports the numerical results, 
	followed by the conclusion in Section~\ref{sec:conclusion}.

	\section{Signal Model of RAQR}\label{sec:signal_model}
	
	\subsection{Four-Level Atomic Model}\label{subsec:electron_transition}
	
	\begin{figure}[t]
		\centering
		\includegraphics[width=\columnwidth]{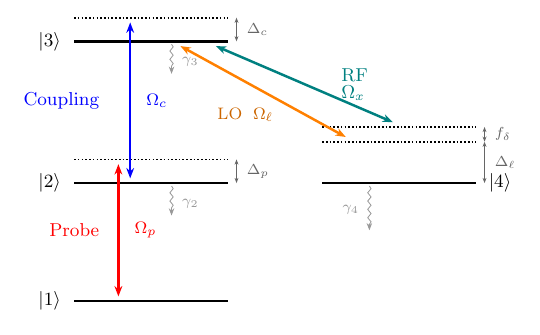}
		\caption{Four-level ladder energy diagram of the RAQR superheterodyne 
			receiver. The probe, coupling, and LO fields drive the 
			successive transitions $|1\rangle{\to}|2\rangle{\to}|3\rangle{\to}|4\rangle$ 
			at detunings $\Delta_p$, $\Delta_c$, and $\Delta_\ell$, 
			respectively, and the weak RF signal $\Omega_x$ sits at offset 
			$f_\delta$ from the LO on the $|3\rangle{\to}|4\rangle$ transition.}
		\label{fig:system}
	\end{figure}
	
	The RAQR is built on a four-level ladder-type atomic system $|1\rangle$--$|2\rangle$--$|3\rangle$--$|4\rangle$, as illustrated in Fig.~\ref{fig:system}. A probe laser with Rabi frequency $\Omega_p$ couples the ground state $|1\rangle$ to the intermediate excited state $|2\rangle$, and a coupling laser with Rabi frequency $\Omega_c$ further drives the atom into the Rydberg state $|3\rangle$. In the superheterodyne configuration, a strong quantum local oscillator~(LO) with Rabi frequency $\Omega_\ell$ drives the adjacent Rydberg transition $|3\rangle \to |4\rangle$, while the desired RF signal sits at a small frequency offset $f_\delta \triangleq f_x - f_\ell$ from the LO. The composite Rabi frequency on the $|3\rangle$--$|4\rangle$ transition is therefore
	\begin{equation}\label{eq:Omega_RF}
		\Omega_{\mathrm{RF}}(t) \approx \Omega_\ell + \Omega_x\cos(2\pi f_\delta t + \varphi_\delta),
	\end{equation}
	where $\Omega_x$ denotes the Rabi-frequency amplitude induced by the incident RF field, and
	$\varphi_\delta$ is the phase offset between the incident RF field and the LO.
	Equation~\eqref{eq:Omega_RF} keeps only the slowly varying envelope at $f_\delta$. The counter-rotating contributions are removed by the rotating-wave approximation that underlies \eqref{eq:Hamiltonian}.
	Since $f_\delta$ is typically much smaller than the dominant atomic decay rate $\gamma_2$ of state $|2\rangle$, the density matrix reaches its steady state well within a single beat period.
	We can therefore treat $\Omega_{\mathrm{RF}}$ as a quasi-static 
	parameter when solving the master equation, and reintroduce 
	its time dependence after extracting the baseband envelope.
	
	Under this approximation, the rotating-frame Hamiltonian of the four-level system is
	\begin{align}\label{eq:Hamiltonian}
		\mathcal{H} =
		\begin{bmatrix}
			0 & \dfrac{\Omega_p}{2} & 0 & 0 \\[4pt]
			\dfrac{\Omega_p}{2} & \Delta_p & \dfrac{\Omega_c}{2} & 0 \\[4pt]
			0 & \dfrac{\Omega_c}{2} & \Delta_p+\Delta_c & \dfrac{\Omega_{\mathrm{RF}}}{2} \\[4pt]
			0 & 0 & \dfrac{\Omega_{\mathrm{RF}}}{2} & \Delta_p+\Delta_c+\Delta_\ell
		\end{bmatrix},
	\end{align}
	where $\Delta_p$, $\Delta_c$, and $\Delta_\ell$ are the detunings of the probe, coupling, and LO fields, respectively. The full atomic dynamics follow the Lindblad master equation $\dot{\bm{\rho}} = -j[\mathcal{H},\bm{\rho}] - \tfrac{1}{2}\{\bm{\Gamma},\bm{\rho}\} + \bm{\Lambda}$, where $\bm{\rho}$ is the atomic density matrix, and the relaxation and repopulation matrices $\bm{\Gamma}$ and $\bm{\Lambda}$ account for spontaneous decay and transit-time broadening.
	In a vapor cell at room temperature, the lifetimes of the 
	Rydberg states $|3\rangle$ and $|4\rangle$ far exceed the 
	baseband period, so $\gamma_2$ is the only non-negligible 
	decoherence rate.
	Imposing $\dot{\bm{\rho}} = 0$ and solving the resulting linear 
	system gives the probe coherence $\rho_{21}$ as a rational 
	function of $\Omega_{\mathrm{RF}}$; the explicit dependence on 
	the detuning parameters is reported in~\cite{gong2025raqr}. 
	The probe absorption is set by $\mathrm{Im}\{\chi\} \propto 
	\mathrm{Im}\{\rho_{21}\}$, so any nonlinearity in $\rho_{21}$ 
	with respect to $\Omega_{\mathrm{RF}}$ is inherited by the 
	optically read-out signal.

\subsection{From RF Field to Detector Voltage}\label{subsec:RF_to_voltage}

Let $u_m$ denote the total RF complex envelope at the $m$-th 
RAQR array element, defined relative to the LO frequency 
$f_\ell$. The superheterodyne down-conversion produces a 
perturbation in the Rabi frequency on the 
$|3\rangle\!\to\!|4\rangle$ transition,
\begin{equation}\label{eq:Delta_Omega}
	\Delta\Omega_m(t) = \frac{\mu_{34}}{\hbar} |u_m| 
	\cos(\omega_\delta t + \angle u_m),
\end{equation}
where $\mu_{34}$ is the $|3\rangle\!\to\!|4\rangle$ 
transition dipole moment and $\omega_\delta = 2\pi f_\delta$. 
The total Rabi frequency on this transition is then 
$\Omega_{\mathrm{RF}}(t) = \Omega_\ell + \Delta\Omega_m(t)$. 
The dependence of $\rho_{21}$ on $\Omega_{\mathrm{RF}}$ is 
carried over to the optical domain via the atomic 
susceptibility
\begin{equation}\label{eq:chi_def}
	\chi(\Omega_{\mathrm{RF}}) = -\frac{2N_0\mu_{12}^2}
	{\varepsilon_0\hbar\Omega_p}\rho_{21}(\Omega_{\mathrm{RF}}),
\end{equation}
where $N_0$ is the atomic number density and $\mu_{12}$ 
the $|1\rangle\!\to\!|2\rangle$ dipole moment. After 
propagating through a vapor cell of length $L_{\mathrm{cell}}$, 
the probe output power and phase become
\begin{align}
	\mathcal{P}_p^{(\mathrm{out})}(\Omega_{\mathrm{RF}}) &= 
	\mathcal{P}_p^{(\mathrm{in})} \exp\!\left(-\frac{2\pi 
		L_{\mathrm{cell}}}{\lambda_p}\Im\{\chi(\Omega_{\mathrm{RF}})\}
	\right), \label{eq:Pp_out}\\
	\phi_p^{(\mathrm{out})}(\Omega_{\mathrm{RF}}) &= 
	\phi_p^{(\mathrm{in})} + \frac{\pi L_{\mathrm{cell}}}
	{\lambda_p}\Re\{\chi(\Omega_{\mathrm{RF}})\}, \label{eq:phi_p_out}
\end{align}
where $\lambda_p$ is the probe wavelength, and 
$\mathcal{P}_p^{(\mathrm{in})}$ and $\phi_p^{(\mathrm{in})}$ 
are the input probe power and phase, respectively.

The modulated probe is then read out by balanced coherent 
optical detection (BCOD): it is interfered with a reference 
beam of power $\mathcal{P}_{\mathrm{IF}}$ and phase 
$\phi_{\mathrm{IF}} = 0$, and the resulting optical signal 
is converted to a voltage by a photodetector followed by 
an LNA~\cite{gong2025raqr}, giving
\begin{equation}\label{eq:VB_def}
	V^{(B)}(\Omega_{\mathrm{RF}}) = \sqrt{G_{\mathrm{LNA}}}\,
	\mathcal{R}_{\mathrm{pd}} \cdot \mathcal{M}(\Omega_{\mathrm{RF}}),
\end{equation}
with photodetector responsivity $\mathcal{R}_{\mathrm{pd}}$, 
LNA gain $G_{\mathrm{LNA}}$, and
\begin{equation}\label{eq:M_def}
	\mathcal{M}(\Omega_{\mathrm{RF}}) \triangleq 2\sqrt{
		\mathcal{P}_{\mathrm{IF}} \cdot \mathcal{P}_p^{(\mathrm{out})}
		(\Omega_{\mathrm{RF}})} \cos\!\big(\phi_p^{(\mathrm{out})}
	(\Omega_{\mathrm{RF}})\big).
\end{equation}
We expand $V^{(B)}$ around the LO operating point 
$\Omega_\ell$ and retain terms up to third order:
\begin{equation}\label{eq:VB_taylor}
	V^{(B)}(\Omega_\ell + \Delta\Omega) \approx a_0 
	+ a_1\Delta\Omega + \tfrac{1}{2}a_2\Delta\Omega^2 
	+ \tfrac{1}{6}a_3\Delta\Omega^3,
\end{equation}
where
\begin{equation}\label{eq:a_coeffs}
	a_n \triangleq \left.\frac{d^n V^{(B)}}{d\Omega^n}
	\right|_{\Omega=\Omega_\ell}, \quad n = 1,2,3,
\end{equation}
and $a_0 \triangleq V^{(B)}(\Omega_\ell)$ is a constant bias 
that is removed by alternating-current~(AC) coupling.

\subsection{Third-Order Baseband Signal Model}\label{subsec:baseband_model}

The existing model in~\cite{gong2025raqr} retains only the 
first-order Taylor coefficient $a_1$, which is accurate 
when the RF field is weak. In a network setting, this is 
no longer sufficient: at high BS density, the aggregate 
interference produces large excursions of 
$\Omega_{\mathrm{RF}}$ around the operating point 
$\Omega_\ell$, and the linear approximation breaks down. 
We therefore carry the Taylor expansion to third order.\footnote{Even-order terms generate direct-current~(DC) and even harmonics, which are rejected by the baseband filter. Higher odd-order terms can also reach the fundamental, but they enter at higher powers of $\Delta\Omega$ around $\Omega_\ell$ and are negligible.}
The cubic term is the leading source of nonlinear 
distortion in this regime, and its strength relative to 
$a_1$ controls how soon the receiver leaves the linear 
regime.

For compactness, define $\tilde{\Omega}_m \triangleq 
(\mu_{34}/\hbar)|u_m|$ and $\varphi_m \triangleq \angle u_m$, 
so that $\Delta\Omega_m(t) = \tilde{\Omega}_m \cos(\omega_\delta 
t + \varphi_m)$. Here $\varphi_m$ absorbs the single-signal 
phase offset $\varphi_\delta$ of \eqref{eq:Omega_RF}. 
Substituting into the Taylor expansion \eqref{eq:VB_taylor} 
yields
\begin{equation}\label{eq:VB_substituted}
	\begin{aligned}
		V^{(B)}(\Omega_\ell + \Delta\Omega_m) = & a_0 + a_1 
		\tilde{\Omega}_m \cos(\omega_\delta t{+}\varphi_m) \\
		&+ \tfrac{1}{2} a_2 \tilde{\Omega}_m^2 \cos^2(\omega_\delta 
		t{+}\varphi_m) \\
		&+ \tfrac{1}{6} a_3 \tilde{\Omega}_m^3 \cos^3(\omega_\delta 
		t{+}\varphi_m).
	\end{aligned}
\end{equation}

After baseband filtering, the zeroth-order term is removed 
by AC coupling, and the second-order term produces only DC 
and $2f_\delta$ components, which lie out of band. The 
remaining contributions at $f_\delta$ come from the first- 
and third-order terms; the detailed derivation is given in 
Appendix~\ref{app:freq_selection}. The third-order term 
contributes to the fundamental frequency with an effective 
coefficient $\tfrac{1}{8}a_3$, which arises from 
$\tfrac{1}{6}a_3 \times \tfrac{3}{4}$. Substituting 
$\tilde{\Omega}_m = (\mu_{34}/\hbar)|u_m|$ and 
$|u_m|e^{j\varphi_m} = u_m$ then gives
\begin{equation}\label{eq:vm_expanded}
	v_m = a_1\frac{\mu_{34}}{\hbar}\,u_m + \frac{1}{8}\,a_3\!
	\left(\frac{\mu_{34}}{\hbar}\right)^{\!3} |u_m|^2\,u_m,
\end{equation}
where $v_m$ is the complex baseband envelope of the 
retained fundamental component.

We then add the photodetection noise 
$w_m \sim \mathcal{CN}(0,\sigma_w^2)$, which lumps together 
shot noise and electronic noise from the BCOD stage and 
enters after the nonlinear transduction. The per-element 
baseband model becomes
\begin{equation}\label{eq:vm_nonlinear}
	v_m = c_1\,u_m + c_3\,|u_m|^2\,u_m + w_m,
\end{equation}
with transduction coefficients
\begin{equation}\label{eq:c1c3}
	c_1 \triangleq a_1 \frac{\mu_{34}}{\hbar}, \qquad
	c_3 \triangleq \frac{1}{8}\,a_3 \!\left(\frac{\mu_{34}}{\hbar}
	\right)^{\!3}.
\end{equation}
Here $c_1$ is the small-signal gain of the Rydberg 
transducer, and recovers the gain coefficient used 
in~\cite{gong2025raqr}. The cubic coefficient $c_3$ 
captures the leading nonlinearity, which is inherited from 
the rational dependence of $\rho_{21}$ on 
$\Omega_{\mathrm{RF}}$. In the weak-field limit 
$|c_3||u_m|^2 \ll |c_1|$, the cubic term is negligible and 
the model reduces to $v_m \approx c_1 u_m + w_m$.

\begin{remark}[Linearity region of the RAQR element]
	\label{rem:linearity}
	In the noiseless case, the relative error of the linear 
	model in \eqref{eq:vm_nonlinear} is
	\begin{equation}\label{eq:lin_error}
		\frac{|v_m - c_1 u_m|}{|c_1 u_m|} = \frac{|c_3|}{|c_1|}\,
		|u_m|^2,
	\end{equation}
	so for a tolerance $\varepsilon$, the linear approximation 
	holds within
	\begin{equation}\label{eq:lin_radius}
		|u_m|^2 \le \varepsilon\,\frac{|c_1|}{|c_3|}.
	\end{equation}
	
	The size of this region is set by $|c_1|/|c_3|$, so the 
	ratio $|c_3/c_1|$ summarizes how soon the per-element 
	response leaves the linear regime. Maximizing $|c_1|$ 
	alone does not enlarge this region, since at typical 
	operating points $|c_3|$ tends to grow together with 
	$|c_1|$. The same ratio also controls the distortion 
	seen at the network level.
\end{remark}
	
\section{Stochastic-Geometry Analysis}\label{sec:SG_framework}

	This section analyzes the uplink coverage probability of a network with RAQR BSs.
	The main difficulty is that the RAQR output is nonlinear, while standard coverage analysis is built around linear signal and interference terms.
	We handle this step by applying the Bussgang decomposition to the per-element RAQR response, which gives an equivalent linear gain and an additive distortion term.

\subsection{Network Model}\label{subsec:network_model}

Consider a single-tier uplink cellular network. 
The BS locations form a homogeneous Poisson point process~(HPPP) $\Phi_b$ with density $\lambda_b$ on $\mathbb{R}^2$, and each BS uses an $M_r$-element RAQR array. 
We consider the typical BS at the origin and denote its serving UE distance by $R_0$. 
To avoid the near-field singularity of the power-law path loss, we condition on $R_0 \ge d_0$ and use $\beta(r)=Cr^{-\alpha}$ for $r\ge d_0$. 
Under nearest-BS association, the conditional density of $R_0$ is
\begin{equation}\label{eq:f_R0}
	f_{R_0}(r) = 2\pi\lambda_b r \exp\!\big(-\pi\lambda_b (r^2 - d_0^2)\big), \quad r \geq d_0,
\end{equation}
which follows by conditioning the standard nearest-neighbor distance distribution on $R_0\ge d_0$.

For tractability, the interfering UEs are approximated by an independent HPPP $\Phi_u^{(\mathrm{int})}$ with density $\lambda_u$ outside the guard zone $B(0,r)$. 
This approximation is commonly used in uplink stochastic geometry analysis with one scheduled UE per cell. 
When one UE is scheduled per active cell, $\lambda_u$ is typically close to $\lambda_b$. 
The desired UE transmits with power $P_0$, and each interfering UE transmits with power $P_u$.

For a UE located at $\mathbf{x}_k \in \mathbb{R}^2$, the narrowband channel to the typical BS is modeled as
\begin{equation}\label{eq:channel_model}
	\mathbf{h}_{0k} = \sqrt{\beta_{0k}}\,\mathbf{g}_{0k},
\end{equation}
where $\beta_{0k}=C\|\mathbf{x}_k\|^{-\alpha}$, $C$ is the path-loss intercept, and $\alpha>2$ is the path-loss exponent. 
The small-scale fading vector satisfies $\mathbf{g}_{0k}\sim\mathcal{CN}(\mathbf{0},\mathbf{I}_{M_r})$. 
The condition $\|\mathbf{x}_i\|\ge r\ge d_0$ keeps all path-loss terms finite.

The RAQR array input is the RF complex envelope vector
\begin{equation}\label{eq:u_vec}
	\mathbf{u} = K\!\left(\sqrt{P_0}\,\mathbf{h}_{00}\,s_0
	+ \sum_{\mathbf{x}_i \in \Phi_u^{(\mathrm{int})}} \sqrt{P_u}\,\mathbf{h}_{0i}\,s_i\right)\!,
\end{equation}
where $s_0$ and $s_i$ are unit-power symbols. 
The coefficient $K \triangleq \sqrt{2/(c\varepsilon_0 A_e)}$ converts the received power envelope to the incident electric-field envelope. Here $c$ is the speed of light, $\varepsilon_0$ is the vacuum permittivity, $f_c$ is the carrier frequency, $\lambda_c = c/f_c$ is the carrier wavelength, and $A_e = \lambda_c^2/(4\pi)$ is the effective aperture. 
The $m$-th entry of $\mathbf{u}$ is the input to the $m$-th RAQR element.

The vector $\mathbf{u}$ contains only the RF signal and interference before optical readout.
The photodetection noise $w_m$ is added after the nonlinear transduction in~\eqref{eq:vm_nonlinear}. 
Thus, the Bussgang decomposition in Section~\ref{subsec:bussgang} is applied to the signal-plus-interference input of the nonlinear block.

	\subsection{Bussgang Linearization and SINR}\label{subsec:bussgang}
		
	From \eqref{eq:vm_nonlinear}, the nonlinear part of each RAQR element can be written as
	$g(u)=c_1u+c_3|u|^2u$, followed by the photodetection noise 
	$w_m\sim\mathcal{CN}(0,\sigma_w^2)$.
	The Bussgang decomposition~\cite{bussgang1952cross,demir2021bussgang} is then applied to $g(u)$.
	This step requires a Gaussian input model for the RAQR element, which is established first.
	For tractability, the Bussgang parameters are computed from the conditional second moment of the RAQR input rather than from each instantaneous PPP realization.
	
	\begin{lemma}[Conditional Input Distribution]\label{lem:input_dist}
		Conditioned on~$R_0 = r$, the per-element input~$u_m$ is well approximated as
		\begin{equation}\label{eq:CSCG_approx}
			u_m \mid r \ \stackrel{\text{approx.}}{\sim}\ \mathcal{CN}\!\left(0,\, \bar{\sigma}_{\mathrm{in}}^2(r)\right),
		\end{equation}
		where the conditional input power is
		\begin{equation}\label{eq:sigma_u_sq}
			\bar{\sigma}_{\mathrm{in}}^2(r) = K^2\!\left(P_0 C\,r^{-\alpha} + \frac{2\pi\lambda_u P_u C}{\alpha - 2}\,r^{2-\alpha}\right).
		\end{equation}
	\end{lemma}
	\begin{proof}
		See Appendix~\ref{app:input_dist_proof}.
	\end{proof}
	
	Note that $\bar{\sigma}_{\mathrm{in}}^2(r)$ does \emph{not} include $\sigma_w^2$, since the noise $w_m$ enters after the nonlinear block.
	
	\begin{lemma}[Bussgang Decomposition of the RAQR]\label{lem:bussgang}
		For $g(u) = c_1 u + c_3|u|^2 u$ with $u \sim \mathcal{CN}(0, \bar{\sigma}_{\mathrm{in}}^2)$, the Bussgang linear gain coefficient and the distortion variance are, respectively,
		\begin{align}
			\kappa(r) &= c_1 + 2c_3\,\bar{\sigma}_{\mathrm{in}}^2(r), \label{eq:kappa}\\
			\sigma_d^2(r) &= 2|c_3|^2\!\left(\bar{\sigma}_{\mathrm{in}}^2(r)\right)^3. \label{eq:sigma_d}
		\end{align}
	\end{lemma}
	\begin{proof}
		See Appendix~\ref{app:bussgang_proof}.
	\end{proof}
	
The coefficient $\kappa(r)$ is the effective linear gain after Bussgang decomposition. 
Compared with $c_1$, it includes an additional term proportional to the aggregate input power $\bar{\sigma}_{\mathrm{in}}^2(r)$. 
Depending on the operating point, this term may reduce the effective gain and produce gain compression. 
The distortion variance $\sigma_d^2(r)$ scales cubically with $\bar{\sigma}_{\mathrm{in}}^2(r)$. 
Thus, the nonlinear distortion is weak at low input power but can become important in dense networks.

Applying Lemma~\ref{lem:bussgang} element-wise, we obtain the equivalent linear model
\begin{equation}\label{eq:v_linear}
	\mathbf{v} = \kappa(r)\mathbf{u}+\mathbf{d}+\mathbf{w},
\end{equation}
where $\mathbf{d}$ is the Bussgang distortion term satisfying
$\mathbb{E}[\mathbf{d}\mathbf{u}^H]=\mathbf{0}$ and
$\mathbb{E}[\mathbf{d}\mathbf{d}^H]=\sigma_d^2(r)\mathbf{I}_{M_r}$.
The photodetection noise is
$\mathbf{w}\sim\mathcal{CN}(\mathbf{0},\sigma_w^2\mathbf{I}_{M_r})$.
This form separates the linear signal and interference terms from the distortion and photodetection noise.

With perfect channel state information~(CSI), the BS applies maximum ratio combining~(MRC), and the decision variable is $\hat{s}_0=\mathbf{h}_{00}^H\mathbf{v}$. 
Substituting \eqref{eq:v_linear} gives
\begin{align}
	\hat{s}_0 &= \underbrace{\kappa K \sqrt{P_0}\,\|\mathbf{h}_{00}\|^2\,s_0}_{\text{desired signal}} 
	+ \underbrace{\kappa K \sum_{\mathbf{x}_i \in \Phi_u^{(\mathrm{int})}} \sqrt{P_u}\,\mathbf{h}_{00}^H \mathbf{h}_{0i}\,s_i}_{\text{interference}} \nonumber\\
	&\quad + \underbrace{\mathbf{h}_{00}^H(\mathbf{d}+\mathbf{w})}_{\text{distortion + noise}}. 
	\label{eq:x_hat_decompose}
\end{align}

We next rewrite the desired and interference powers using scalar fading variables. 
Let
\begin{equation}\label{eq:H0_def}
	H_0 \triangleq \|\mathbf{g}_{00}\|^2 \sim \Gamma(M_r,1),
\end{equation}
so that $\|\mathbf{h}_{00}\|^2=\beta_{00}H_0$ with $\beta_{00}=Cr^{-\alpha}$. 
Define the unit-norm direction $\mathbf{t}_0=\mathbf{h}_{00}/\|\mathbf{h}_{00}\|$. 
For the cross-channel term,
\begin{equation}\label{eq:cross_channel}
	|\mathbf{h}_{00}^H\mathbf{h}_{0i}|^2
	= \|\mathbf{h}_{00}\|^2 |\mathbf{t}_0^H\mathbf{h}_{0i}|^2
	= \|\mathbf{h}_{00}\|^2\beta_{0i}G_i,
\end{equation}
where $G_i\triangleq|\mathbf{t}_0^H\mathbf{g}_{0i}|^2\sim\exp(1)$ due to the isotropy of $\mathbf{g}_{0i}$ and its independence from $\mathbf{t}_0$.

Define the aggregate interference power as
\begin{equation}\label{eq:I_def}
	I(r) \triangleq \sum_{\mathbf{x}_i \in \Phi_u^{(\mathrm{int})}} P_u\beta_{0i}G_i,
\end{equation}
where the dependence on $r$ comes from the exclusion region $\|\mathbf{x}_i\|\ge r$.

Dividing the signal, interference, and distortion-plus-noise powers by $|\kappa(r)|^2K^2\|\mathbf{h}_{00}\|^2$ gives the following result.

\begin{theorem}[SINR under RAQR with Third-Order Nonlinearity]\label{thm:SINR}
	The post-MRC SINR at the typical RAQR BS, conditioned on the serving distance $R_0=r$, is
	\begin{equation}\label{eq:SINR_final}
		\mathrm{SINR}_0
		= \frac{P_0\beta_{00}H_0}{I(r)+\sigma_{\mathrm{eff}}^2(r)},
	\end{equation}
	where $H_0\sim\Gamma(M_r,1)$ and the effective input-referred noise power is
	\begin{equation}\label{eq:sigma_eff}
		\sigma_{\mathrm{eff}}^2(r)
		= \frac{\sigma_w^2+\sigma_d^2(r)}{|\kappa(r)|^2K^2}.
	\end{equation}
\end{theorem}
	
	\begin{proof}
		For brevity, let us denote $\kappa=\kappa(r)$ and $\sigma_d^2=\sigma_d^2(r)$. Starting from \eqref{eq:x_hat_decompose}, the desired signal power is
		\begin{equation}
			S = |\kappa|^2 K^2 P_0\,\|\mathbf{h}_{00}\|^4 = |\kappa|^2 K^2 P_0\,\beta_{00}^2 H_0^2.
		\end{equation}
		The interference power, conditioned on the PPP realization, is
		\begin{align}
			\begin{aligned}
				I_{\mathrm{tot}} & = |\kappa|^2 K^2 \sum_{\mathbf{x}_i \in \Phi_u^{(\mathrm{int})}} P_u\,|\mathbf{h}_{00}^H \mathbf{h}_{0i}|^2 \\
				&= |\kappa|^2 K^2\,\|\mathbf{h}_{00}\|^2\,I(r),
			\end{aligned}
		\end{align}
		where the final step uses the definition in \eqref{eq:I_def}.
		The noise-plus-distortion power is
		\begin{equation}
			N_{\mathrm{tot}} = \mathbb{E}\!\left[|\mathbf{h}_{00}^H(\mathbf{d}+\mathbf{w})|^2 \,\big|\, \mathbf{h}_{00}\right] = (\sigma_d^2 + \sigma_w^2)\,\|\mathbf{h}_{00}\|^2.
		\end{equation}
		Therefore,
		\begin{align}
			\mathrm{SINR}_0 &= \frac{S}{I_{\mathrm{tot}} + N_{\mathrm{tot}}} \nonumber \\
			&= \frac{|\kappa|^2 K^2 P_0\,\beta_{00}^2 H_0^2}{|\kappa|^2 K^2\,\beta_{00} H_0\,I(r) + (\sigma_d^2+\sigma_w^2)\beta_{00} H_0}.
		\end{align}
		Dividing numerator and denominator by $|\kappa|^2K^2\beta_{00}H_0$ yields \eqref{eq:SINR_final}.
	\end{proof}

	\subsection{Coverage Probability}\label{subsec:coverage}

	We first characterize the Laplace transform of the aggregate interference, then derive the conditional and unconditional coverage probabilities.
	
	\begin{lemma}[Laplace Transform of the Aggregate Interference]\label{lem:laplace}
		The conditional Laplace transform of $I(r)$, averaging over fading and PPP locations, is
		\begin{equation}\label{eq:laplace_general}
			\mathcal{L}_{I|R_0}(s|r) = \exp\!\left(-2\pi\lambda_u \int_{r}^{\infty} \frac{sP_uC\,t^{-\alpha}}{1 + sP_uC\,t^{-\alpha}}\,t\,dt\right).
		\end{equation}
		For $\alpha = 4$, this evaluates in closed form to
		\begin{equation}\label{eq:laplace_closed}
			\mathcal{L}_{I|R_0}(s|r) = \exp\!\left(-\pi\lambda_u\sqrt{\zeta}\left(\frac{\pi}{2} - \arctan\!\left(\frac{r^2}{\sqrt{\zeta}}\right)\right)\right),
		\end{equation}
		where $\zeta \triangleq sP_uC$.
	\end{lemma}
	\begin{proof}
		See Appendix~\ref{app:coverage_proofs}.
	\end{proof}
	
	Let $\theta > 0$ denote the SINR coverage threshold. Conditioned on the serving distance $R_0 = r$, the coverage event $\{\mathrm{SINR}_0 > \theta\}$ from \eqref{eq:SINR_final} can be rewritten as
	\begin{equation}\label{eq:coverage_event}
		H_0 > s(r,\theta)\!\left(I(r) + \sigma_{\mathrm{eff}}^2(r)\right),
	\end{equation}
	where the normalized threshold is defined as
	\begin{equation}\label{eq:s_def}
		s(r,\theta) \triangleq \frac{\theta}{P_0\,\beta_{00}} = \frac{\theta\,r^{\alpha}}{P_0\,C}.
	\end{equation}
	
	Since $H_0 \sim \Gamma(M_r, 1)$ with integer shape parameter $M_r$, its exact CCDF is $\mathbb{P}[H_0 > x] = e^{-x}\sum_{n=0}^{M_r-1} x^n/n!$. Substituting this sum form into the coverage integral requires evaluating successive derivatives of $\mathcal{L}_{I|R_0}$, which is analytically cumbersome for the $\alpha = 4$ case, since the $\arctan$ structure in \eqref{eq:laplace_closed} makes higher-order derivatives progressively unwieldy as $M_r$ grows. We therefore adopt the following proposition which allows the coverage probability to be evaluated directly through the Laplace transform.
	
	\begin{proposition}[Gamma-CCDF Approximation]\label{prop:gamma_approx}
		The CCDF of $H_0 \sim \Gamma(M_r, 1)$ is approximated as
		\begin{equation}\label{eq:Gamma_approx}
			\mathbb{P}[H_0 > x] \approx 1 - \left(1 - e^{-\nu x}\right)^{M_r},
		\end{equation}
		where $\nu = \alpha_{\mathrm{ks}}/M_r$ minimizes the Kolmogorov--Smirnov~(K-S) distance to the true $\Gamma(M_r, 1)$ cumulative distribution function (CDF):
		\begin{equation}\label{eq:alpha_ks}
			\alpha_{\mathrm{ks}} = \arg\min_{\xi > 0} \sup_{x \geq 0} \left| F_{\bar{H}_0}(x) - \left(1 - e^{-\xi x}\right)^{M_r} \right|,
		\end{equation}
		with $\bar{H}_0 = H_0/M_r \sim \Gamma(M_r, 1/M_r)$. For $M_r = 1$, $\alpha_{\mathrm{ks}} = 1$ and the approximation is exact. 
		The classical Alzer constant $(M_r!)^{-1/M_r}$~\cite{Alzer1997} provides an analytical upper bound; however, it is not optimal under the K-S distance criterion adopted here.
	\end{proposition}
	
	Since $H_0$ and $I(r)$ are conditionally independent given $r$ as $H_0$ depends on $\|\mathbf{g}_{00}\|$ while $I(r)$ depends on the independent direction $\mathbf{t}_0$ and $\sigma_{\mathrm{eff}}^2(r)$ is deterministic conditioned on $r$, expanding \eqref{eq:Gamma_approx} via the binomial theorem yields the following result.
	
	\begin{theorem}[Conditional Coverage Probability]\label{thm:cond_coverage}
		The coverage probability conditioned on $R_0 = r$ is
		\begin{equation}\label{eq:pcov_cond}
			p_{\mathrm{cov}}(\theta|r) \approx \sum_{k=1}^{M_r}\binom{M_r}{k}(-1)^{k+1} e^{-k\nu\,s\,\sigma_{\mathrm{eff}}^2(r)}\,\mathcal{L}_{I|R_0}\!\big(k\nu\,s\big|r\big),
		\end{equation}
		where $s \equiv s(r,\theta) = \theta\,r^{\alpha}/(P_0 C)$, $\sigma_{\mathrm{eff}}^2(r)$ is given by \eqref{eq:sigma_eff}, $\nu = \alpha_{\mathrm{ks}}/M_r$ with $\alpha_{\mathrm{ks}}$ defined in \eqref{eq:alpha_ks}, and $\mathcal{L}_{I|R_0}(\cdot|r)$ is given by \eqref{eq:laplace_closed}. When $M_r = 1$, the approximation is exact and reduces to $p_{\mathrm{cov}}(\theta|r) = e^{-s\,\sigma_{\mathrm{eff}}^2(r)}\,\mathcal{L}_{I|R_0}(s|r)$.
	\end{theorem}
	\begin{proof}
		See Appendix~\ref{app:coverage_proofs}.
	\end{proof}
	
	Averaging over the serving distance with the conditional PDF \eqref{eq:f_R0} yields the network-level coverage probability.
	
	\begin{theorem}[Network Coverage Probability]\label{thm:uncond_coverage}
		The coverage probability conditioned on $R_0 \geq d_0$ is
		\begin{equation}\label{eq:pcov_final}
			p_{\mathrm{cov}}(\theta) = \int_{d_0}^{\infty} p_{\mathrm{cov}}(\theta|r)\cdot f_{R_0}(r)\,dr,
		\end{equation}
		where $p_{\mathrm{cov}}(\theta|r)$ is given by \eqref{eq:pcov_cond} and $f_{R_0}(r)$ is defined in \eqref{eq:f_R0}.
	\end{theorem}
	
	Substituting \eqref{eq:pcov_cond} into \eqref{eq:pcov_final} and interchanging the sum and integral:
	\begin{align}\label{eq:pcov_expanded}
		p_{\mathrm{cov}}(\theta) = &\sum_{k=1}^{M_r}\binom{M_r}{k}(-1)^{k+1} \nonumber\\
		&\times \int_{d_0}^{\infty} e^{-k\nu\,s\,\sigma_{\mathrm{eff}}^2(r)}\,\mathcal{L}_{I|R_0}(k\nu\,s\,|\,r)\,f_{R_0}(r)\,dr.
	\end{align}

	\section{Special Cases and Extensions}\label{sec:special_cases}
	
	In this section, we specialize the general coverage analysis of Section~\ref{sec:SG_framework} to two limiting cases that reveal the fundamental tradeoffs of RAQR-based networks, and extend the analysis to account for coupling-induced receive correlation in conventional antenna arrays.
	
	\subsection{Conventional Receiver Benchmark}\label{subsec:conv_benchmark}
	
	To quantify the coverage gain offered by the RAQR, we first establish a conventional receiver baseline. The conventional linear receiver corresponds to the special case $c_3 = 0$ (no nonlinear transduction) with fixed thermal noise $\sigma_n^2 = k_B T_{\mathrm{sys}} F B$, where $k_B$ is the Boltzmann constant, $T_{\mathrm{sys}}$ is the system noise temperature, $F$ is the receiver noise figure, and $B$ is the signal bandwidth.
	
	\begin{corollary}[Conventional Receiver SINR]\label{cor:sinr_conv}
		Setting $c_3 = 0$ in Theorem~\ref{thm:SINR} yields $\kappa = c_1$, $\sigma_d^2 = 0$, and
		\begin{equation}\label{eq:SINR_conv}
			\mathrm{SINR}_0^{(\mathrm{conv})} = \frac{P_0\,\beta_{00}\,H_0}{\displaystyle\sum_{\mathbf{x}_i \in \Phi_u^{(\mathrm{int})}} P_u\,\beta_{0i}\,G_i + \sigma_n^2},
		\end{equation}
		where $H_0 \sim \Gamma(M_r, 1)$ and $G_i \sim \exp(1)$.
	\end{corollary}
	
	\begin{corollary}[Conventional Receiver Coverage]\label{cor:pcov_conv}
		The conditional coverage probability of the conventional receiver is
		\begin{multline}\label{eq:pcov_conv}
			p_{\mathrm{cov}}^{(\mathrm{conv})}(\theta|r)
			= \sum_{k=1}^{M_r}\binom{M_r}{k}(-1)^{k+1}
			e^{-k\nu\,s(r,\theta)\,\sigma_n^2}\\
			\times\mathcal{L}_{I|R_0}\!\big(k\nu\,s(r,\theta)\big|r\big),
		\end{multline}
		where $\mathcal{L}_{I|R_0}(\cdot|r)$ is given by \eqref{eq:laplace_closed}.
	\end{corollary}
	
	Comparing \eqref{eq:pcov_conv} with the RAQR coverage probability in \eqref{eq:pcov_cond}, the two expressions share the same Laplace transform $\mathcal{L}_{I|R_0}$ and differ only in the noise exponent: the conventional receiver uses the constant $\sigma_n^2$, while the RAQR uses the distance-dependent $\sigma_{\mathrm{eff}}^2(r)$. This difference captures the core tradeoff studied in this paper: at short serving distances, the Bussgang distortion inflates $\sigma_{\mathrm{eff}}^2(r)$ well above $\sigma_n^2$, penalizing the RAQR; at large distances, the distortion vanishes and $\sigma_{\mathrm{eff}}^2(r)$ drops far below $\sigma_n^2$, giving the RAQR an advantage.

	\subsection{Noise-Limited Regime}\label{subsec:noise_limited}
	
	When $\lambda_u \to 0$, the aggregate interference vanishes and $\mathcal{L}_{I|R_0}(s|r) \to 1$. Define the nonlinearity ratio
	\begin{equation}\label{eq:rho_NL}
		\rho_{\mathrm{NL}}(r) \triangleq \frac{|c_3|\,\bar{\sigma}_{\mathrm{in}}^2(r)}{|c_1|},
	\end{equation}
	which quantifies the relative strength of the third-order distortion to the linear gain. In the noise-limited regime, $\bar{\sigma}_{\mathrm{in}}^2(r)$ is dominated by the desired signal term $K^2 P_0 C\,r^{-\alpha}$, and $\rho_{\mathrm{NL}}(r) \ll 1$ for $r \geq d_0$, so the RAQR operates in the linear regime with $\kappa(r) \approx c_1$ and $\sigma_d^2(r) \approx 0$. Define the RAQR intrinsic noise power
	\begin{equation}\label{eq:sigma_RAQR}
		\sigma_{\mathrm{RAQR}}^2 \triangleq \frac{\sigma_w^2}{|c_1|^2\,K^2}.
	\end{equation}
	
	\begin{corollary}[Noise-Limited Coverage]\label{cor:noise_lim}
		In the noise-limited regime, the conditional coverage probability \eqref{eq:pcov_cond} reduces to
		\begin{equation}\label{eq:pcov_noise_lim}
			p_{\mathrm{cov}}^{(\mathrm{noise})}(\theta|r)
			= \sum_{k=1}^{M_r}\binom{M_r}{k}(-1)^{k+1}
			e^{-k\nu\,s(r,\theta)\,\sigma_{\mathrm{eff}}^2},
		\end{equation}
		where $\sigma_{\mathrm{eff}}^2 = \sigma_{\mathrm{RAQR}}^2$ for the RAQR and $\sigma_{\mathrm{eff}}^2 = \sigma_n^2$ for the conventional receiver.
	\end{corollary}
	
	\begin{proposition}[RAQR Advantage in the Noise-Limited Regime]\label{prop:noise_lim_ratio}
		For $M_r = 1$, the coverage-gain ratio is
		\begin{equation}\label{eq:cov_ratio_noise}
			\frac{p_{\mathrm{cov}}^{(\mathrm{RAQR})}(\theta|r)}{p_{\mathrm{cov}}^{(\mathrm{conv})}(\theta|r)}
			= \exp\!\left(s(r,\theta)\left(\sigma_n^2 - \sigma_{\mathrm{RAQR}}^2\right)\right) > 1.
		\end{equation}
		
		For general $M_r$, the inequality $p_{\mathrm{cov}}^{(\mathrm{RAQR})} > p_{\mathrm{cov}}^{(\mathrm{conv})}$ holds since every term in the sum \eqref{eq:pcov_noise_lim} satisfies $e^{-k\nu s\,\sigma_{\mathrm{RAQR}}^2} > e^{-k\nu s\,\sigma_n^2}$.
	\end{proposition}
	
	Under typical parameters, $\sigma_{\mathrm{RAQR}}^2/\sigma_n^2 \approx 10^{-6}$, so the RAQR's quantum-limited noise floor provides a significant coverage advantage that grows exponentially with the SINR threshold $\theta$ and the path loss $r^{\alpha}$.

	
	\subsection{Extension to Coupling-Induced Receive Correlation}\label{subsec:coupling}
	
	In conventional antenna arrays, mutual coupling and the induced spatial correlation become more pronounced when antenna elements are densely packed. To capture this effect while maintaining analytical tractability, we model the receive-side coupling by a positive-definite matrix $\mathbf{R} \in \mathbb{C}^{M_r \times M_r}$ and adopt the exponential form
	\begin{equation}\label{eq:R_exp}
		[\mathbf{R}]_{ij} = \varrho^{|i-j|}, \qquad \varrho \in [0,1),
	\end{equation}
	which is widely used as a tractable model for spatially correlated multi-antenna channels~\cite{sanguinetti2020toward}. The effective channel is then written as
	\begin{equation}\label{eq:h_coupled}
		\tilde{\mathbf{h}}_{0k} = \sqrt{\beta_{0k}}\,\mathbf{R}^{1/2}\,\mathbf{g}_{0k}, \quad \mathbf{g}_{0k} \sim \mathcal{CN}(\mathbf{0},\mathbf{I}_{M_r}).
	\end{equation}
	
	Under MRC, the desired-channel gain becomes
	\begin{equation}\label{eq:H_tilde}
		\tilde{H}_0 \triangleq \mathbf{g}_{00}^H\,\mathbf{R}\,\mathbf{g}_{00} = \sum_{k=1}^{M_r} \lambda_k\,|g'_k|^2,
	\end{equation}
	where $\{\lambda_k\}$ are the eigenvalues of $\mathbf{R}$ and $|g'_k|^2 \sim \exp(1)$ are i.i.d. This eigenvalue-domain representation is standard for correlated MRC channels~\cite{mckay2007mimomrc}. Stochastic-geometry studies have likewise shown that correlation-aware modeling is essential when MRC operates in PPP interference fields, rather than under branch-wise independence~\cite{tanbourgi2014effect}. By contrast, Rydberg atom-based receiver arrays can be idealized as coupling-free in the classical antenna sense, and are therefore modeled by the benchmark case $\mathbf{R} = \mathbf{I}$~\cite{yuan2025mimo}.
	
	Following the same MRC derivation as in Theorem~\ref{thm:SINR}, the post-combining SINR under correlation is
	\begin{equation}\label{eq:SINR_coup}
		\mathrm{SINR}_0^{(\mathrm{coup})} = \frac{P_0\,\beta_{00}\,\tilde{H}_0}{\displaystyle\sum_{\mathbf{x}_i \in \Phi_u^{(\mathrm{int})}} P_u\,\beta_{0i}\,\tilde{G}_i + \sigma_n^2},
	\end{equation}
	where $\tilde{H}_0$ is given by \eqref{eq:H_tilde}. When $\mathbf{R} = \mathbf{I}$, all $\lambda_k = 1$ and $\tilde{H}_0 \sim \Gamma(M_r,1)$, recovering the uncoupled case. For $\varrho > 0$, the unequal eigenvalues increase the variance of $\tilde{H}_0$ from $M_r$ to $\mathrm{tr}(\mathbf{R}^2) > M_r$ while preserving the mean $\mathbb{E}[\tilde{H}_0] = M_r$; the resulting dispersion in the effective channel gain worsens the coverage performance and reduces the effective array gain under strong correlation.
	
	Conditioned on the MRC direction $\tilde{\mathbf{t}}_0$, the interference projection $\tilde{G}_i = |\tilde{\mathbf{t}}_0^H \mathbf{R}^{1/2}\mathbf{g}_{0i}|^2$ follows $\exp(\mu_R)$. Here $\tilde{\mathbf{t}}_0 = \mathbf{R}^{1/2}\mathbf{g}_{00}/\|\mathbf{R}^{1/2}\mathbf{g}_{00}\|$ and $\mu_R = \tilde{\mathbf{t}}_0^H\mathbf{R}\,\tilde{\mathbf{t}}_0$. This random scale is coupled to $\tilde{H}_0$ through $\mathbf{g}_{00}$. To keep the analysis tractable, we replace $\mu_R$ by its deterministic surrogate
	\begin{equation}\label{eq:mu_I}
		\mu_I \triangleq \frac{\mathrm{tr}(\mathbf{R}^2)}{\mathrm{tr}(\mathbf{R})} = \frac{\sum_{k=1}^{M_r}\lambda_k^2}{M_r},
	\end{equation}	
	so that $\tilde{G}_i \approx \exp(\mu_I)$. Note that the desired-signal channel gain $\tilde{H}_0$ is computed exactly from \eqref{eq:H_tilde}. Here, replacing $\mu_R$ with $\mu_I$ only affects the interference term and preserves the mean of the interference projection, since $\mathbb{E}[\mu_R] = \mu_I$.
	
	Under this approximation, the interference Laplace transform becomes
	\begin{equation}\label{eq:laplace_coup}
		\mathcal{L}_{I|R_0}^{(\mathrm{coup})}(s|r) \approx \exp\!\left(-\pi\lambda_u\sqrt{\zeta_c}\left(\frac{\pi}{2} - \arctan\!\frac{r^2}{\sqrt{\zeta_c}}\right)\right),
	\end{equation}
	where $\zeta_c \triangleq sP_uC\mu_I$, which reduces to the uncoupled Laplace \eqref{eq:laplace_closed} when $\mu_I = 1$.
	
	When the eigenvalues of $\mathbf{R}$ are distinct (which holds for $\varrho > 0$), the CCDF of $\tilde{H}_0$ admits the partial-fraction expansion
	\begin{equation}\label{eq:CCDF_Htilde}
		\mathbb{P}[\tilde{H}_0 > x] = \sum_{k=1}^{M_r} \omega_k\,e^{-x/\lambda_k}, \quad \omega_k = \prod_{j \neq k}\frac{\lambda_k}{\lambda_k - \lambda_j}.
	\end{equation}
	
	\begin{proposition}[Coverage under Coupling-Induced Correlation]\label{prop:pcov_coup}
		Under the approximation $\tilde{G}_i \approx \exp(\mu_I)$, the conditional coverage probability with receive correlation $\mathbf{R}$ is
		\begin{multline}\label{eq:pcov_coup}
			p_{\mathrm{cov}}^{(\mathrm{coup})}(\theta|r) \approx \sum_{k=1}^{M_r} \omega_k\,\exp\!\left(-\frac{s(r,\theta)\,\sigma_n^2}{\lambda_k}\right) \\
			\times \mathcal{L}_{I|R_0}^{(\mathrm{coup})}\!\left(\frac{s(r,\theta)}{\lambda_k}\,\bigg|\,r\right),
		\end{multline}
		where $\omega_k$ and $\lambda_k$ are defined in \eqref{eq:CCDF_Htilde}, $\mathcal{L}_{I|R_0}^{(\mathrm{coup})}$ is given by \eqref{eq:laplace_coup}, and $s(r,\theta) = \theta r^{\alpha}/(P_0 C)$.
	\end{proposition}
	\begin{proof}
		See Appendix~\ref{app:coupling_proof}.
	\end{proof}
	
	When $\varrho = 0$, we have $\lambda_k = 1$ and $\mu_I = 1$ for all $k$, so $\tilde{H}_0 \sim \Gamma(M_r,1)$ and $\tilde{G}_i \sim \exp(1)$, recovering the exact uncoupled coverage probability. Corollary~\ref{cor:pcov_conv} provides a K-S approximation of this same quantity. Since the RAQR's coupling-free operation ($\mathbf{R} = \mathbf{I}$) always achieves $\varrho = 0$, it avoids the correlation-induced channel-gain dispersion entirely, thereby preserving the full array gain associated with $M_r$ independent diversity branches.

	\section{Simulation Results}\label{sec:simulation}
	
	The default system parameters are listed in Table~\ref{tab:params}. All Monte Carlo results are averaged over $N_{\mathrm{mc}} = 30{,}000$ independent trials. Unless stated otherwise, we use $^{133}$Cs as the default atomic species with the optimized detunings given in the table.
	
	\begin{table}[t]
		\centering
		\caption{Default Simulation Parameters}
		\label{tab:params}
		\renewcommand{\arraystretch}{1.15}
		\begin{tabular}{l l l}
			\toprule
			\textbf{Parameter} & \textbf{Symbol} & \textbf{Value} \\
			\midrule
			\multicolumn{3}{l}{\textit{Network}} \\
			Carrier frequency         & $f_c$            & $10$\,GHz \\
			Path-loss exponent        & $\alpha$         & $4$ \\
			Desired / interferer power & $P_0,\,P_u$    & $1$\,W \\
			Minimum serving distance  & $d_0$            & $10$\,m \\
			Number of RAQR elements   & $M_r$            & $10$ \\
			Signal bandwidth          & $B$              & $1$\,MHz \\
			Photodetection noise      & $\sigma_w^2$     & $2\times10^{-9}$\,V$^2$ \\
			System noise temperature  & $T_{\mathrm{sys}}$ & $290$\,K \\
			Receiver noise figure     & $F$              & $5$\,dB \\
			\midrule
			\multicolumn{3}{l}{\textit{Atomic ($^{133}$Cs, $42D_{5/2}\!\to\!43P_{3/2}$)}} \\
			Probe wavelength          & $\lambda_p$      & $852$\,nm \\
			Decay rate of $|2\rangle$ & $\gamma_2$       & $2\pi \times 5.2$\,MHz \\
			Probe detuning (opt.)     & $\Delta_p$       & $2\pi \times (-5.97)$\,MHz \\
			Coupling detuning (opt.)  & $\Delta_c$       & $2\pi \times 4.97$\,MHz \\
			LO detuning (opt.)        & $\Delta_\ell$    & $2\pi \times 0.15$\,MHz \\
			Linear gain               & $|c_1|$          & $0.327$ \\
			Nonlinearity ratio        & $|c_3/c_1|$      & $281$ \\
			\midrule
			\multicolumn{3}{l}{\textit{Simulation}} \\
			Monte Carlo trials        & $N_{\mathrm{mc}}$ & $30{,}000$ \\
			SINR threshold range      & $\theta$         & $-10$ to $20$\,dB \\
			\bottomrule
		\end{tabular}
	\end{table}

	\subsection{Detuning Optimization and Parameter Landscape}\label{sec:sim_a}
	
	\begin{figure*}[h]
		\centering
		\includegraphics[width=1\textwidth]{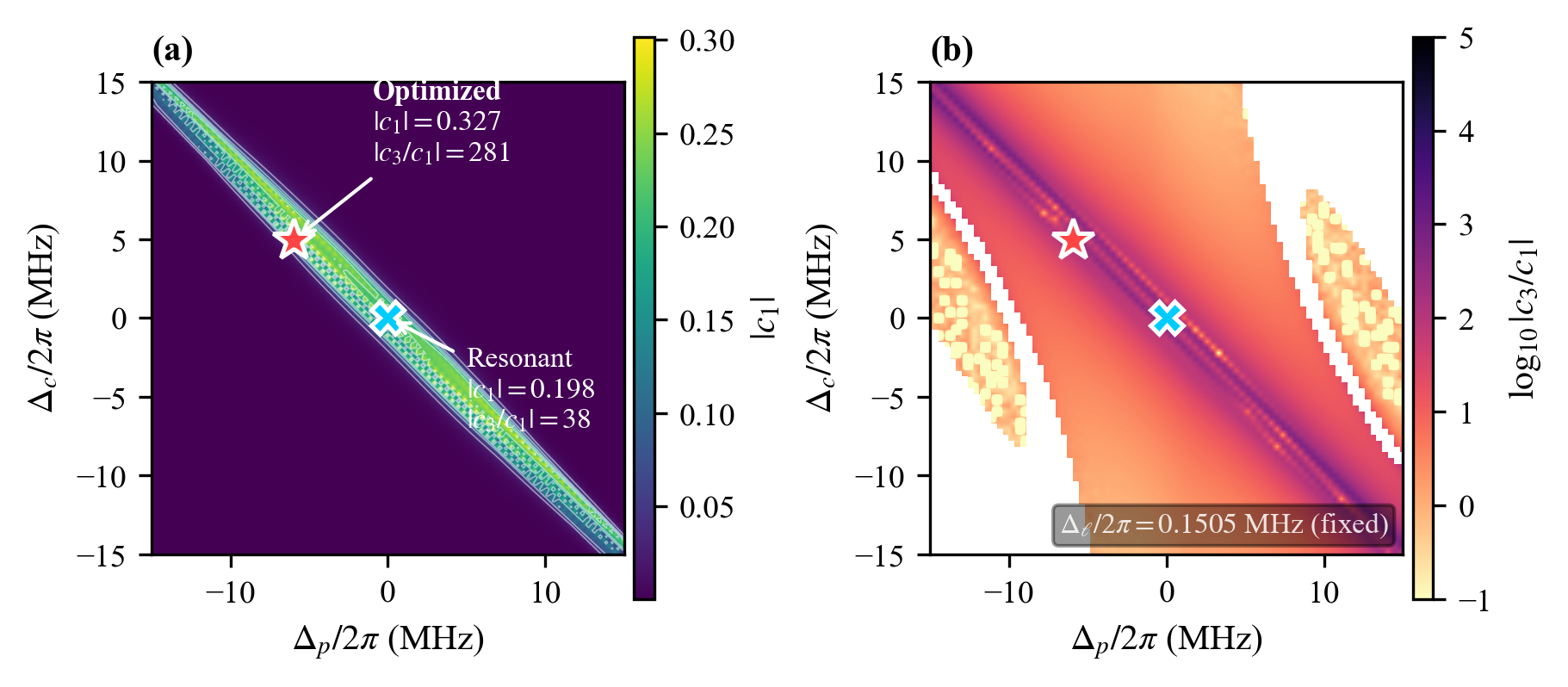}
		\caption{Detuning landscape for the default $^{133}$Cs configuration. Here $\Delta_\ell$ is fixed at its optimized value, and the remaining parameters follow Table~\ref{tab:params}. (a)~Linear transduction gain $|c_1|$. (b)~Nonlinearity ratio $|c_3/c_1|$ on a logarithmic scale. The star marks the optimized point, and the cross marks $(\Delta_p,\Delta_c)=(0,0)$ with the same $\Delta_\ell$.}
		\label{fig:fig_sim_a}
	\end{figure*}
	
	With the atomic species and Rydberg state pair fixed, the remaining design parameters are the laser detunings $(\Delta_p, \Delta_c, \Delta_\ell)$. They determine $c_1$ and $c_3$ through the signal chain of Section~\ref{sec:signal_model}.
	We follow~\cite{wu2023ldr} and maximize the small-signal gain $|c_1|$:
	\begin{equation}\label{eq:detuning_opt}
		(\Delta_p^*, \Delta_c^*, \Delta_\ell^*)
		= \arg\max_{(\Delta_p, \Delta_c, \Delta_\ell)} |c_1(\Delta_p, \Delta_c, \Delta_\ell)|.
	\end{equation}
	We solve~\eqref{eq:detuning_opt} by multi-start local search, summarized in Algorithm~\ref{alg:detuning_opt}.
	
	Fig.~\ref{fig:fig_sim_a} shows $|c_1|$ and $|c_3/c_1|$ over the $(\Delta_p, \Delta_c)$ plane for $^{133}$Cs. The high-$|c_1|$ band lies along $\Delta_p + \Delta_c \approx 0$, where the two-photon resonance enhances the probe coherence. The same band also has large $|c_3/c_1|$. The optimized point (star) reaches $|c_1| = 0.327$ but $|c_3/c_1| = 281$, while the resonant point $\Delta_p = \Delta_c = 0$ has a lower $|c_1| = 0.198$ and a much smaller $|c_3/c_1| = 38$. The two ratios are coupled because $\rho_{21}$ is a rational function of $\Omega_{\mathrm{RF}}$. Operating points with a steeper linear slope also enter gain compression sooner, so $|c_3|$ grows faster than $|c_1|$. At high interference, $|c_3/c_1|$ controls the Bussgang distortion, so the $|c_1|$-optimal detunings are not the best choice for network coverage.

	\begin{algorithm}[t]
		\caption{Multi-Start Detuning Optimization}
		\label{alg:detuning_opt}
		\begin{algorithmic}[1]
			\Require Atomic species and Rydberg state pair; search bounds $[\Delta_{\min},\Delta_{\max}]$; number of random starts $N_s$
			\Ensure Optimized detunings $(\Delta_p^*,\Delta_c^*,\Delta_\ell^*)$; coefficients $c_1^*$, $c_3^*$
			\For{$n = 1,\ldots,N_s$}
			\State Sample $(\Delta_p^{(0)},\Delta_c^{(0)},\Delta_\ell^{(0)})$ from $[\Delta_{\min},\Delta_{\max}]^3$
			\State Run Nelder--Mead to maximize $|c_1(\Delta_p,\Delta_c,\Delta_\ell)|$ using the signal model in Section~\ref{sec:signal_model}
			\State Record local optimum $(\Delta_p^{(n)},\Delta_c^{(n)},\Delta_\ell^{(n)})$ and $|c_1^{(n)}|$
			\EndFor
			\State $(\Delta_p^*,\Delta_c^*,\Delta_\ell^*) \gets \arg\max_n |c_1^{(n)}|$
			\State Compute $c_1^*$ and $c_3^*$ at the optimized point via \eqref{eq:c1c3}
		\end{algorithmic}
	\end{algorithm}
	
\subsection{Analytical Validation and Density Tradeoff}
\label{sec:sim_b}

\begin{figure}[h]
	\centering
	\includegraphics[width=\linewidth]{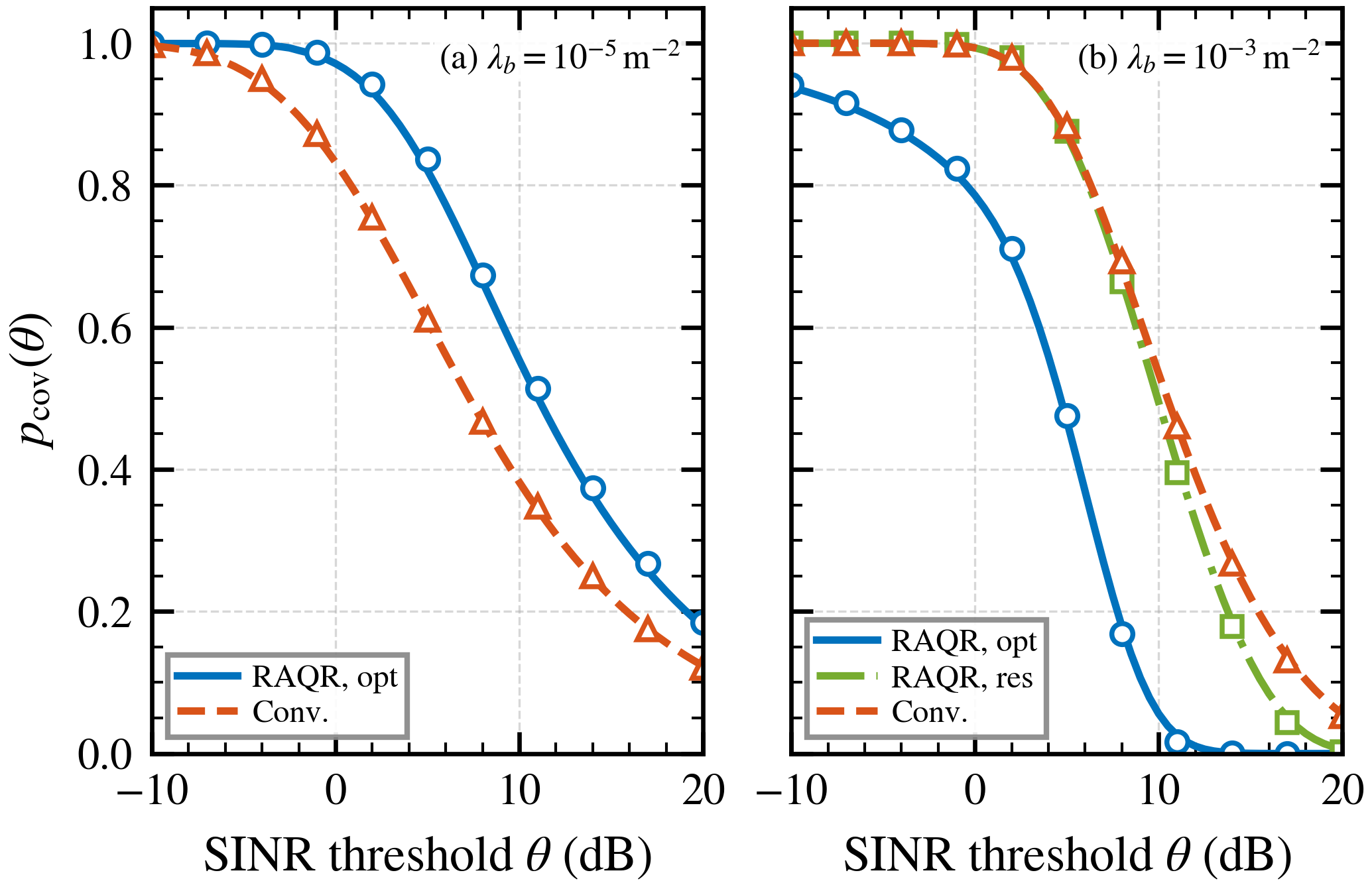}\\[-1mm]
	\includegraphics[width=\linewidth]{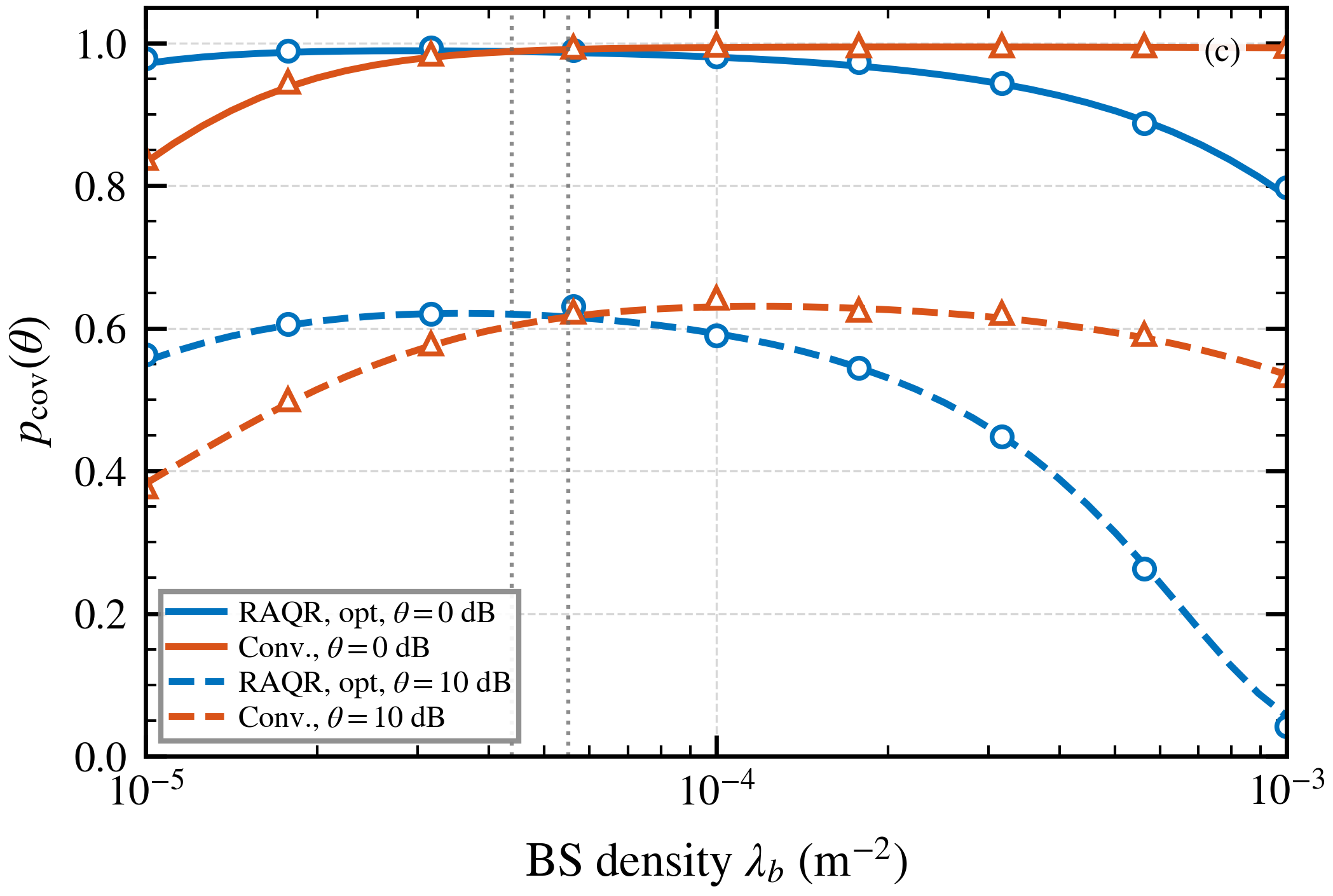}
	\caption{Coverage probability under the BS-density tradeoff with
		$M_r=10$, $\lambda_u=\lambda_b$, and the remaining parameters taken
		from Table~\ref{tab:params}. Lines denote analytical results, and
		open markers denote Monte Carlo simulations. (a)~Coverage versus
		SINR threshold $\theta$ at $\lambda_b=10^{-5}\;\mathrm{m}^{-2}$.
		(b)~Coverage versus $\theta$ at
		$\lambda_b=10^{-3}\;\mathrm{m}^{-2}$, including the resonant RAQR
		with all detunings set to zero. (c)~Coverage versus BS density
		$\lambda_b$ for $\theta=0$ and $10$~dB.}
	\label{fig:fig_sim_b}
\end{figure}

Fig.~\ref{fig:fig_sim_b} compares the coverage probability of the
optimized Cs RAQR and the conventional receiver in both the
threshold and density domains.  The analytical curves from
Theorems~\ref{thm:cond_coverage} and~\ref{thm:uncond_coverage}
follow the Monte~Carlo results in all panels. The largest absolute
difference between analysis and Monte~Carlo over the plotted curves
is about $0.015$, or $1.5$ percentage points.

At $\lambda_b=10^{-5}\;\mathrm{m}^{-2}$, the RAQR has higher coverage
over the plotted $\theta$ range. In this sparse regime, aggregate
interference is weak, and the low RAQR noise floor
$\sigma_{\mathrm{RAQR}}^{2} \ll \sigma_n^{2}$ dominates the noise budget.
As $\lambda_b$ increases, stronger aggregate input raises
$\bar{\sigma}_{\mathrm{in}}^{2}$, and the Bussgang distortion from
the cubic term grows substantially. At
$\lambda_b=10^{-3}\;\mathrm{m}^{-2}$, this distortion dominates
$\sigma_{\mathrm{eff}}^{2}$, and the optimized RAQR falls below the
conventional receiver. The resonant RAQR in
Fig.~\ref{fig:fig_sim_b}(b), with all detunings set to zero, has a
lower linear gain ($|c_1|=0.194$) but a much smaller nonlinearity
ratio ($|c_3/c_1|=38$), so it outperforms the optimized RAQR at this density.

Fig.~\ref{fig:fig_sim_b}(c) also plots the coverage as a function of
$\lambda_b$, with $\lambda_u=\lambda_b$. For $\theta=0$ and $10$~dB,
the optimized RAQR attains higher coverage than the conventional receiver at
low BS densities, while the advantage disappears as the network becomes
denser. The two crossover points occur at
$\lambda_b\approx4.4\times10^{-5}\,\mathrm{m}^{-2}$ and
$\lambda_b\approx5.5\times10^{-5}\,\mathrm{m}^{-2}$, respectively.
Beyond these points, the RAQR coverage falls steeply with $\lambda_b$,
because Bussgang distortion scales as $\bar{\sigma}_{\mathrm{in}}^{6}$,
i.e., cubically in $\lambda_b$, while the conventional noise grows only linearly.
These crossovers indicate that the operating point should balance
small-signal sensitivity with receiver linearity.

	\subsection{Comparison Between Cs and Rb}\label{sec:sim_c}
	
	The coefficients $c_1$ and $c_3$ depend on the atomic species and the
	selected Rydberg transition. To check how this choice changes the
	coverage tradeoff, we compare optimized $^{133}$Cs with optimized
	$^{85}$Rb ($59D_{5/2}\!\to\!60P_{3/2}$, $\lambda_p = 780$\,nm).
	After detuning optimization, Rb gives $|c_1| = 0.192$ and
	$|c_3/c_1| = 81$. The corresponding Cs values are
	$|c_1| = 0.327$ and $|c_3/c_1| = 281$. Thus Cs has a
	$1.7\times$ larger linear gain, but its nonlinearity ratio is also
	about $3.5\times$ larger. The optimized Rb gain is close to that of
	the resonant Cs point in Section~\ref{sec:sim_b}, although the two
	cases still differ in $|c_3/c_1|$ ($81$ for Rb and $38$ for resonant
	Cs).
	
	Fig.~\ref{fig:fig_sim_c}(a) shows the normalized $V^{(B)}$ curves
	and their first order and third order Taylor approximations. The Cs
	curve has a steeper slope near the operating point, which is
	consistent with its larger $|c_1|$. 
	It also leaves the linear
	approximation sooner, indicating a stronger nonlinear response. The
	Rb curve is shallower, but it stays close to its linear approximation
	over a wider input range. 
	The two shaded bands show the typical range of
	$\Omega_{\mathrm{RF}}/\Omega_l$ caused by aggregate interference around
	the operating point. The inter site distance (ISD) denotes the typical
	spacing between neighboring BSs and is approximated as
	$\mathrm{ISD}\approx1/\sqrt{\lambda_b}$.
	At $\lambda_b = 10^{-5}\;\mathrm{m}^{-2}$, corresponding to
	$\mathrm{ISD}\approx316$ m, this range is narrow and remains close to the
	operating point, so both species are effectively linear. At
	$\lambda_b = 10^{-3}\;\mathrm{m}^{-2}$, corresponding to
	$\mathrm{ISD}\approx32$ m, this range is much wider.

	At $\lambda_b = 10^{-5}\;\mathrm{m}^{-2}$
	(Fig.~\ref{fig:fig_sim_c}(b)), Bussgang distortion is small, and the
	linear gain mainly determines the coverage. Cs is above Rb over the
	plotted $\theta$ range. At $\theta = 5$\,dB, the coverage
	probabilities are $0.82$ for Cs, $0.73$ for Rb, and $0.61$ for the
	conventional receiver. Both RAQR species are above the conventional
	receiver in this sparse regime. At
	$\lambda_b = 10^{-3}\;\mathrm{m}^{-2}$
	(Fig.~\ref{fig:fig_sim_c}(c)), the ranking changes. The larger
	$|c_3/c_1|$ of Cs makes its coverage degrade first. At
	$\theta = 5$\,dB, the Cs coverage is $0.47$, while Rb reaches $0.73$
	and the conventional receiver reaches $0.87$. Rb loses roughly
	$40\%$ of the Cs linear gain, but its smaller nonlinearity ratio gives
	higher coverage in the dense regime. It still remains below the
	linear conventional receiver.
	
	The species comparison therefore changes the same two quantities as
	the detuning choice, namely $|c_1|$ and $|c_3/c_1|$.
	Changing the atomic species is not only a way to increase the receiver
	gain. It also changes the curvature of the atomic response around the
	operating point. A species or transition with a lower $|c_1|$ can still
	give better network coverage if it keeps the cubic term small over the
	input range induced by interference.
	
	\begin{figure}[t]
		\centering
		\includegraphics[width=\linewidth]{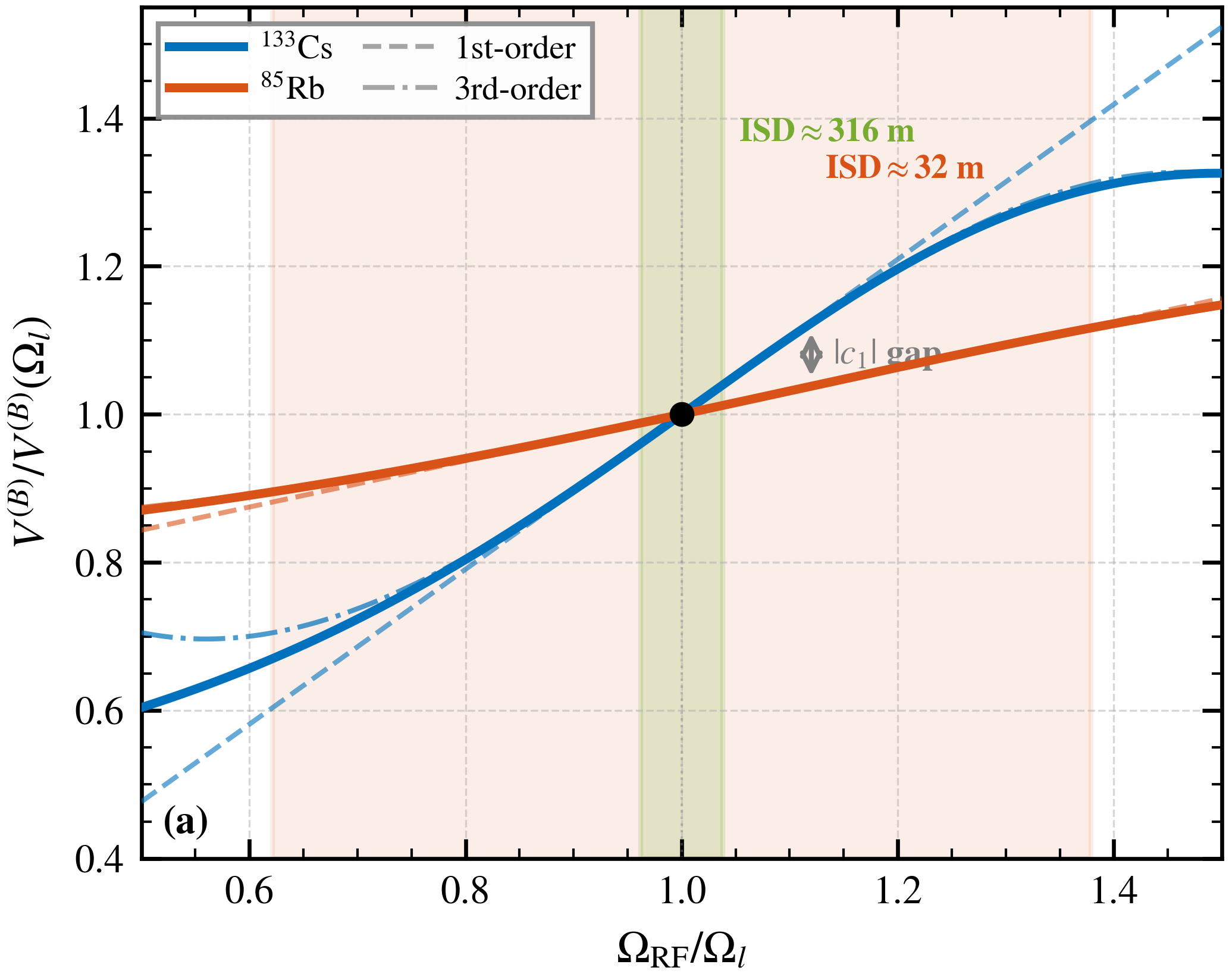}\\[-1mm]
		\includegraphics[width=\linewidth]{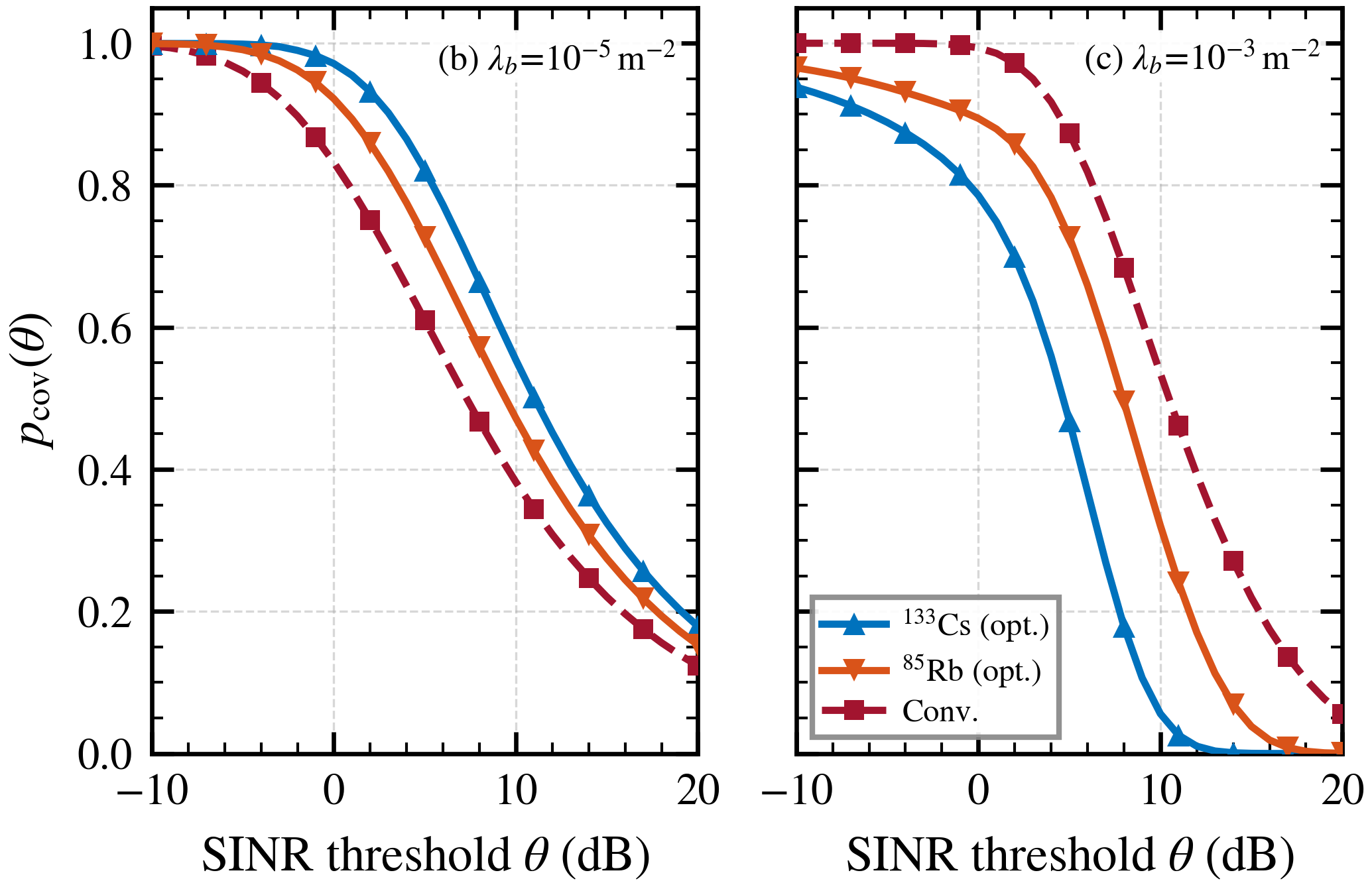}
		\caption{Comparison between $^{133}$Cs and $^{85}$Rb under the default system settings in Table~\ref{tab:params}. (a)~Normalized $V^{(B)}$ curves with first order and third order approximations. 
		The shaded bands indicate the typical range of
		$\Omega_{\mathrm{RF}}/\Omega_l$ caused by aggregate interference for
		$\lambda_b = 10^{-5}$ and $10^{-3}\;\mathrm{m}^{-2}$, corresponding to
		ISD $\approx316$ m and $32$ m, respectively.
		(b)~Coverage probability versus $\theta$ at $\lambda_b = 10^{-5}\;\mathrm{m}^{-2}$. (c)~Coverage probability versus $\theta$ at $\lambda_b = 10^{-3}\;\mathrm{m}^{-2}$.}
		\label{fig:fig_sim_c}
	\end{figure}
	
\subsection{Array Scaling and Nonlinearity Cost}\label{sec:sim_d}

Fig.~\ref{fig:fig_sim_d} studies whether increasing the number of receive elements can offset the nonlinear loss of the optimized RAQR. We fix a moderate BS density, $\lambda_b = 2\times10^{-4}\;\mathrm{m}^{-2}$, and sweep $M_r$ for several SINR thresholds $\theta$.

For both receivers, the coverage probability increases monotonically with $M_r$. The gap between the RAQR and the conventional receiver also narrows as $M_r$ grows, especially for lower SINR thresholds. 
Thus increasing the array size can offset part of the nonlinear loss of the RAQR. 
However, for $M_r$ from $1$ to $32$, the optimized RAQR remains below the conventional receiver. At the lowest threshold, $\theta=0$~dB, the two curves become nearly indistinguishable around $M_r \approx 24$.

The amount of compensation depends strongly on $\theta$. At $\theta = 0$~dB, the gap shrinks to zero by $M_r \approx 24$. By contrast, at $\theta = 10$~dB, the conventional receiver still leads by about $5.3\%$ even at $M_r = 32$. The number of elements needed to compensate the nonlinear loss therefore grows quickly as the target SINR increases.

This is the same $|c_3/c_1|$ effect seen in Sections~\ref{sec:sim_b} and~\ref{sec:sim_c}. Increasing $M_r$ enhances coherent combining for both receivers, but in the RAQR case part of that gain is absorbed by the distortion term rather than as coverage gain. Adding more elements help resolve this, but does not change the operating point. For high thresholds, the receiver still needs a more linear operating point or a different atomic configuration with a smaller $|c_3/c_1|$.

\begin{figure}[t]
	\centering
	\includegraphics[width=\linewidth]{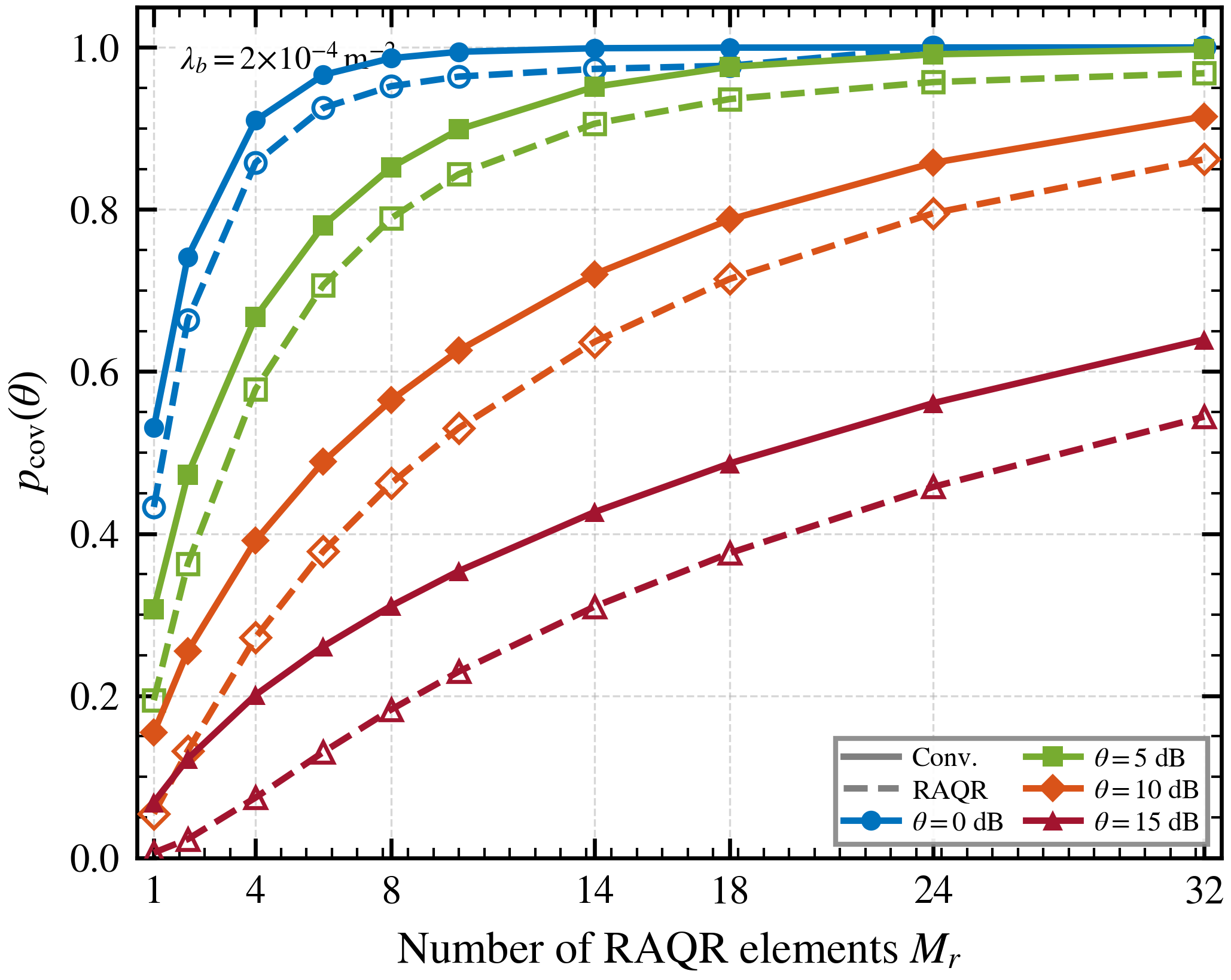}
	\caption{Coverage probability versus the number of RAQR array elements $M_r$ at $\lambda_b = 2\times10^{-4}\;\mathrm{m}^{-2}$ for four SINR thresholds, with the remaining parameters taken from Table~\ref{tab:params}. Solid lines denote the conventional receiver, and dashed lines denote the optimized RAQR.}
	\label{fig:fig_sim_d}
\end{figure}

	\subsection{Receive Correlation and Array Scaling}\label{sec:sim_e}
	
	Fig.~\ref{fig:fig_sim_e} evaluates the coupling model of Section~\ref{subsec:coupling}. In Fig.~\ref{fig:fig_sim_e}(a), we fix $M_r=10$ and set $\lambda_b=\lambda_u=10^{-4}\;\mathrm{m}^{-2}$. We compare the RAQR at the uncoupled point $\mathbf{R}=\mathbf{I}$ with conventional arrays using receive correlation coefficients $\varrho=0.5$ and $\varrho=0.7$. The analytical curves follow the Monte Carlo markers closely. The largest absolute difference between analysis and Monte~Carlo, taken over the three plotted curve pairs and the full $\theta$ grid, is $0.0364$, or about $3.64$ percentage points.
	
	As $\varrho$ increases, the conventional receiver loses coverage. The unequal eigenvalues of $\mathbf{R}$ widen the effective channel gain distribution and raise the interference factor $\mu_I$. Both effects reduce the usable array gain. The RAQR does not incur this correlation penalty because its elements are modeled with $\mathbf{R}=\mathbf{I}$. It stays above the $\varrho=0.7$ conventional array over almost all plotted thresholds, and it is also above the $\varrho=0.5$ case except where all curves are already close to one.
	
	Fig.~\ref{fig:fig_sim_e}(b) fixes $\theta=5$~dB and plots the outage probability $1-p_{\mathrm{cov}}(\theta)$ versus $M_r$ on a logarithmic scale. The log scale avoids the ceiling effect that would hide differences once the coverage curves approach one. 
	All outage curves decrease with $M_r$. The uncoupled conventional benchmark gives the lowest outage because it has neither coupling nor Bussgang distortion. The RAQR outage decreases with a slope close to this uncoupled benchmark, while the coupled conventional curves decrease more slowly. For $\varrho=0.7$, the outage remains several times larger than that of the RAQR at the largest plotted $M_r$. Array scaling helps the optimized RAQR, but coupling can reduce the array gain of a conventional receiver as $M_r$ grows.
	
	\begin{figure}[t]
		\centering
		\includegraphics[width=\linewidth]{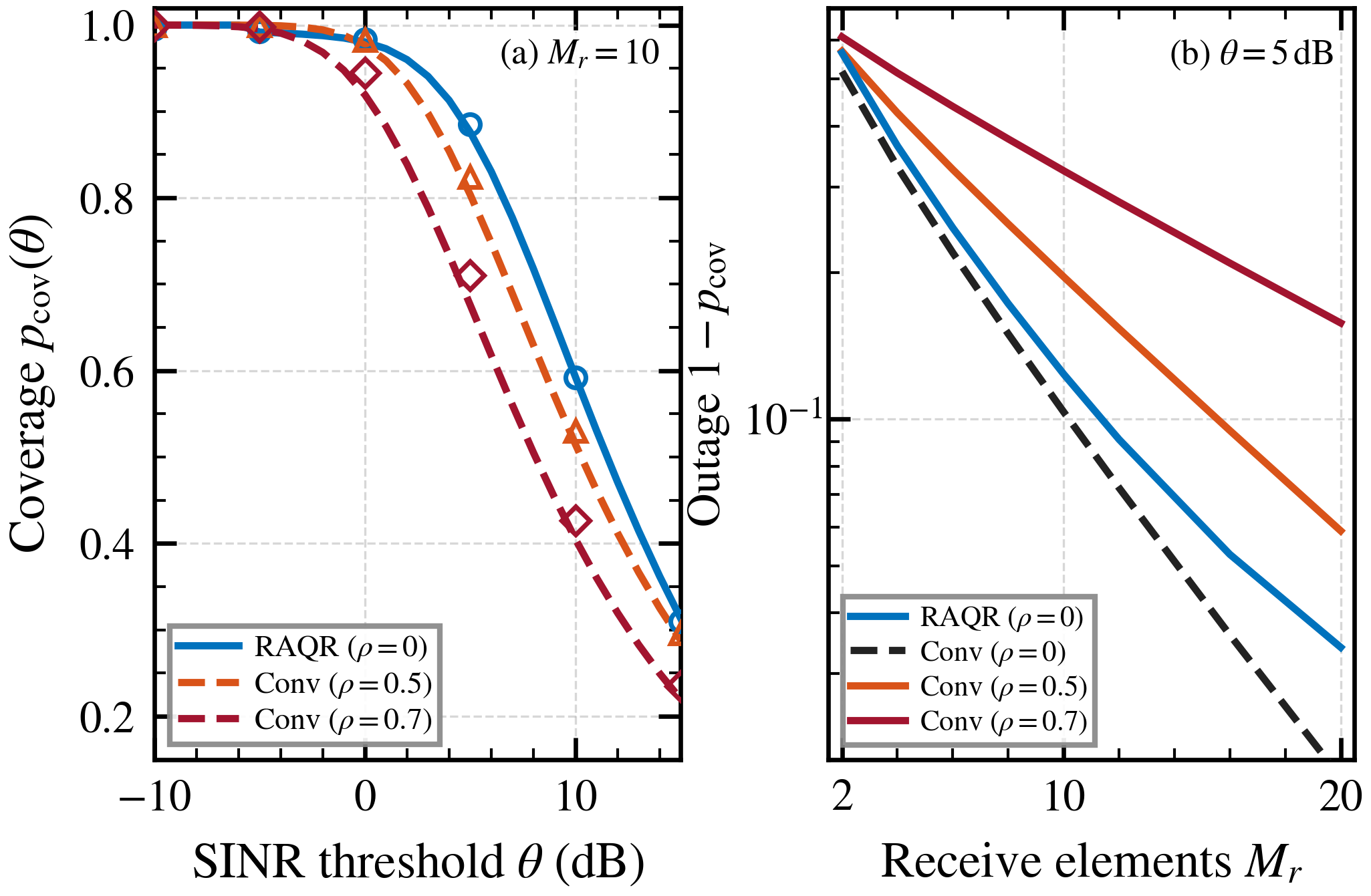}
		\caption{Performance under receive correlation with the remaining parameters taken from Table~\ref{tab:params}. (a)~Coverage probability versus SINR threshold $\theta$ at $\lambda_b=\lambda_u=10^{-4}\;\mathrm{m}^{-2}$ and $M_r=10$. The RAQR is evaluated at $\mathbf{R}=\mathbf{I}$, and the conventional receiver is evaluated with $\varrho \in \{0.5,\,0.7\}$. Lines denote analysis, and markers denote Monte Carlo simulations. (b)~Outage probability $1-p_{\mathrm{cov}}(\theta)$ versus the number of receive elements $M_r$ at $\theta=5$~dB and $\lambda_b=\lambda_u=10^{-4}\;\mathrm{m}^{-2}$. The outage plot in (b) uses analytical curves only.}
		\label{fig:fig_sim_e}
	\end{figure}
	
	\section{Conclusion}\label{sec:conclusion}
	
This paper has studied uplink coverage in cellular networks with RAQR arrays.
By combining the third-order RAQR baseband model with the Bussgang decomposition, the coverage analysis captures how input power changes the receiver's effective noise and distortion.
The simulations show that the RAQR advantage depends on the network regime. In sparse or noise-limited networks, RAQR arrays benefit from their low noise floor and achieve higher coverage than conventional receivers. In denser networks, stronger aggregate input power increases nonlinear distortion and can remove this advantage. Therefore, the operating point should balance small-signal sensitivity with receiver linearity.
A larger array can reduce part of this loss, and the weak coupling of RAQR elements helps retain array gain when conventional antenna arrays suffer from receive correlation.

	\appendices
	
	\section{Frequency Selection in the Third-Order Expansion}\label{app:freq_selection}
	
	This appendix details the spectral analysis of each Taylor order in \eqref{eq:VB_substituted}, justifying the baseband signal obtained in Section~\ref{subsec:baseband_model}. Let $\phi \triangleq \omega_\delta t + \varphi_m$ for brevity.
	
	\emph{Zeroth and first orders.}
	The DC term $a_0$ is removed by AC coupling. The first-order term $a_1 \tilde{\Omega}_m \cos\phi$ lies at $f_\delta$ and passes directly to baseband.
	
	\emph{Second order.}
	Applying $\cos^2\phi = \tfrac{1}{2}(1 + \cos 2\phi)$:
	\begin{equation}\label{eq:second_order}
		\tfrac{1}{2} a_2 \tilde{\Omega}_m^2 \cos^2\phi
		= \tfrac{1}{4} a_2 \tilde{\Omega}_m^2
		+ \tfrac{1}{4} a_2 \tilde{\Omega}_m^2 \cos 2\phi.
	\end{equation}
	Both the DC and $2f_\delta$ components fall outside the baseband filter, so the entire second-order contribution is rejected.
	
	\emph{Third order.}
	Applying $\cos^3\phi = \tfrac{3}{4}\cos\phi + \tfrac{1}{4}\cos 3\phi$:
	\begin{equation}\label{eq:third_order}
		\tfrac{1}{6} a_3 \tilde{\Omega}_m^3 \cos^3\phi
		= \tfrac{1}{8} a_3 \tilde{\Omega}_m^3 \cos\phi
		+ \tfrac{1}{24} a_3 \tilde{\Omega}_m^3 \cos 3\phi.
	\end{equation}
	The $3f_\delta$ harmonic is filtered out, but the fundamental-frequency component survives with an effective coefficient $\tfrac{1}{8} = \tfrac{1}{6} \times \tfrac{3}{4}$.
	
	\emph{Baseband envelope.}
	Collecting the surviving terms at $f_\delta$, the difference-frequency signal is
	\begin{equation}\label{eq:V_fdelta}
		V_{f_\delta}(t)
		= \bigl( a_1 \tilde{\Omega}_m + \tfrac{1}{8} a_3 \tilde{\Omega}_m^3 \bigr) \cos\phi.
	\end{equation}
	After standard in-phase/quadrature (IQ) down-conversion, the complex envelope becomes $v_m = (a_1 \tilde{\Omega}_m + \tfrac{1}{8} a_3 \tilde{\Omega}_m^3) e^{j\varphi_m}$. Substituting $\tilde{\Omega}_m = (\mu_{34}/\hbar)|u_m|$ and using $|u_m|^k e^{j\varphi_m} = |u_m|^{k-1} u_m$ recovers~\eqref{eq:vm_expanded}.

	\section{Proof of Lemma~\ref{lem:input_dist}}\label{app:input_dist_proof}
	
	The $m$-th element input can be written as
	\begin{equation}\label{eq:um_expand}
		u_m = K\!\left(\sqrt{P_0}\,h_{00,m}\,s_0
		+ \sum_{\mathbf{x}_i\in\Phi_u^{(\mathrm{int})}} \sqrt{P_u}\,h_{0i,m}\,s_i \right)\!,
	\end{equation}
	where $h_{00,m}$ and $h_{0i,m}$ denote the $m$-th entries of $\mathbf{h}_{00}$ and $\mathbf{h}_{0i}$, respectively. We verify that $u_m \mid r$ is well approximated as proper circularly symmetric complex Gaussian (CSCG) by checking its first- and second-order moments.
	
	\emph{Zero mean.} Each fading entry $h_{00,m}$, $h_{0i,m}$ is zero-mean CSCG, and the symbols $s_0, s_i$ are independent of the fading with $\mathbb{E}[s_0] = \mathbb{E}[s_i] = 0$. By independence,
	\begin{align}\label{eq:um_zero_mean}
		\mathbb{E}[u_m \mid r]
		&= K\!\bigg(\sqrt{P_0}\,\mathbb{E}[h_{00,m}]\mathbb{E}[s_0] \nonumber\\
		&\quad + \sum_i \sqrt{P_u}\,\mathbb{E}[h_{0i,m}]\mathbb{E}[s_i]\bigg) = 0.
	\end{align}
	
	\emph{Proper (vanishing pseudo-variance).} Recall that a zero-mean complex random variable $u$ is proper if $\mathbb{E}[u^2] = 0$. For the desired term, $h_{00,m}^2 s_0^2$ has $\mathbb{E}[h_{00,m}^2] = 0$ since $h_{00,m}$ is CSCG. The same holds for every interferer. Cross terms between distinct indices vanish by independence. Hence
	\begin{equation}\label{eq:um_proper}
		\mathbb{E}[u_m^2 \mid r] = 0.
	\end{equation}
	
	\emph{Second moment.} The conditional variance is
	\begin{align}\label{eq:var_expand}
		\mathbb{E}[|u_m|^2 \mid r]
		&= K^2\!\left(P_0 C\,r^{-\alpha}
		+ P_u C\,\mathbb{E}\!\bigg[\sum_{\mathbf{x}_i\in\Phi_u^{(\mathrm{int})}}
		\|\mathbf{x}_i\|^{-\alpha} \,\bigg|\, r\bigg]\right)\!.
	\end{align}
	Applying Campbell's theorem to the HPPP $\Phi_u^{(\mathrm{int})}$ with exclusion region $\|\mathbf{x}_i\| \geq r$,
	\begin{equation}\label{eq:campbell}
		\mathbb{E}\!\bigg[\sum_{\mathbf{x}_i\in\Phi_u^{(\mathrm{int})}}
		\|\mathbf{x}_i\|^{-\alpha} \,\bigg|\, r\bigg]
		= 2\pi\lambda_u \int_r^\infty t^{1-\alpha}\,dt
		= \frac{2\pi\lambda_u}{\alpha-2}\,r^{2-\alpha}.
	\end{equation}
	Substitute this result into (65) yields (21).
	
	Since $u_m$ is a sum of independent zero-mean proper complex random variables with finite variance, the central limit theorem for triangular arrays ensures that, as the mean number of dominant interferers grows, $u_m \mid r$ converges in distribution to $\mathcal{CN}(0, \bar{\sigma}_{\mathrm{in}}^2(r))$. The three moment conditions above fix the parameters of this limiting distribution. \hfill$\blacksquare$

	\section{Proof of Lemma~\ref{lem:bussgang}}\label{app:bussgang_proof}
	
	For $u \sim \mathcal{CN}(0,\bar{\sigma}_{\mathrm{in}}^2)$, the complex Gaussian moments are
	\begin{equation}\label{eq:CSCG_moments}
		\mathbb{E}[|u|^2] = \bar{\sigma}_{\mathrm{in}}^2, \quad \mathbb{E}[|u|^4] = 2(\bar{\sigma}_{\mathrm{in}}^2)^2, \quad \mathbb{E}[|u|^6] = 6(\bar{\sigma}_{\mathrm{in}}^2)^3.
	\end{equation}
	
	\emph{Bussgang gain.} The Bussgang decomposition seeks the scalar $\kappa$ that minimizes $\mathbb{E}[|g(u) - \kappa u|^2]$. Setting $\partial/\partial\kappa^* = 0$ yields $\kappa = \mathbb{E}[g(u) u^*]/\mathbb{E}[|u|^2]$. Substituting $g(u) = c_1 u + c_3|u|^2 u$,
	\begin{equation}\label{eq:kappa_derive}
		\begin{aligned}
			\kappa &= \frac{\mathbb{E}[(c_1 u + c_3|u|^2 u)u^*]}{\bar{\sigma}_{\mathrm{in}}^2} \\
			&= \frac{c_1\bar{\sigma}_{\mathrm{in}}^2 + 2c_3(\bar{\sigma}_{\mathrm{in}}^2)^2}{\bar{\sigma}_{\mathrm{in}}^2} \\
			&= c_1 + 2c_3\bar{\sigma}_{\mathrm{in}}^2.
		\end{aligned}
	\end{equation}
	
	\emph{Residual orthogonality.} Let $d \triangleq g(u) - \kappa u = c_3(|u|^2 - 2\bar{\sigma}_{\mathrm{in}}^2)u$. By construction of $\kappa$,
	\begin{equation}\label{eq:orthogonality}
		\mathbb{E}[d\,u^*] = \mathbb{E}[g(u)u^*] - \kappa\,\mathbb{E}[|u|^2] = \kappa\bar{\sigma}_{\mathrm{in}}^2 - \kappa\bar{\sigma}_{\mathrm{in}}^2 = 0,
	\end{equation}
	so $d$ is uncorrelated with $u$. This justifies treating $\mathbf{d}$ as an independent additive term in \eqref{eq:v_linear} for the subsequent second-order interference analysis.
	
	\emph{Distortion variance.}
	\begin{align}
		\sigma_d^2 &= \mathbb{E}[|d|^2] = |c_3|^2\,\mathbb{E}\!\left[\left||u|^2 - 2\bar{\sigma}_{\mathrm{in}}^2\right|^2 |u|^2\right] \nonumber\\
		&= |c_3|^2\!\left(\mathbb{E}[|u|^6] - 4\bar{\sigma}_{\mathrm{in}}^2\,\mathbb{E}[|u|^4] + 4(\bar{\sigma}_{\mathrm{in}}^2)^2\,\mathbb{E}[|u|^2]\right) \nonumber\\
		&= |c_3|^2\!\left(6(\bar{\sigma}_{\mathrm{in}}^2)^3 - 8(\bar{\sigma}_{\mathrm{in}}^2)^3 + 4(\bar{\sigma}_{\mathrm{in}}^2)^3\right) \\
		&= 2|c_3|^2(\bar{\sigma}_{\mathrm{in}}^2)^3,
	\end{align}
	which completes the proof. \hfill$\blacksquare$

	\section{Proofs of Lemma~\ref{lem:laplace} and Theorem~\ref{thm:cond_coverage}}\label{app:coverage_proofs}
	
	\emph{Proof of Lemma~\ref{lem:laplace}.}
	Starting from the definition \eqref{eq:I_def},
	\begin{align}
		\mathcal{L}_{I|R_0}(s|r)
		&= \mathbb{E}\!\left[\exp\!\left(-s \sum_{\mathbf{x}_i \in \Phi_u^{(\mathrm{int})}} P_u C\,\|\mathbf{x}_i\|^{-\alpha}\,G_i\right)\right] \nonumber\\
		&= \mathbb{E}_{\Phi_u}\!\left[\prod_{\mathbf{x}_i \in \Phi_u^{(\mathrm{int})}} \mathbb{E}_{G_i}\!\left[e^{-s P_u C\,\|\mathbf{x}_i\|^{-\alpha}\,G_i}\right]\right]. \label{eq:laplace_step1}
	\end{align}
	Since $G_i \sim \exp(1)$, its moment generating function gives
	\begin{equation}\label{eq:mgf_exp}
		\mathbb{E}_{G_i}\!\left[e^{-s P_u C\,\|\mathbf{x}_i\|^{-\alpha}\,G_i}\right] = \frac{1}{1 + s P_u C\,\|\mathbf{x}_i\|^{-\alpha}}.
	\end{equation}
	Applying the PGFL of the HPPP $\Phi_u^{(\mathrm{int})}$ with density $\lambda_u$ on $\{\mathbf{x} \in \mathbb{R}^2 : \|\mathbf{x}\| \geq r\}$,
	\begin{align}
		\mathcal{L}_{I|R_0}(s|r)
		&= \exp\!\left(-\lambda_u \int_{\|\mathbf{x}\|\geq r}\!\!\left(1 - \frac{1}{1 + s P_u C\,\|\mathbf{x}\|^{-\alpha}}\right)\!d\mathbf{x}\right) \nonumber\\
		&= \exp\!\left(-2\pi\lambda_u \int_r^\infty \frac{s P_u C\,t^{-\alpha}}{1 + s P_u C\,t^{-\alpha}}\,t\,dt\right), \label{eq:laplace_general_proof}
	\end{align}
	which matches \eqref{eq:laplace_general}. For $\alpha = 4$, substitute $\tau = t^2$ ($d\tau = 2t\,dt$):
	\begin{align}
		2\pi\lambda_u \int_r^\infty \frac{\zeta\,t^{-4}}{1 + \zeta\,t^{-4}}\,t\,dt
		&= \pi\lambda_u \int_{r^2}^\infty \frac{\zeta}{\tau^2 + \zeta}\,d\tau \nonumber\\
		&= \pi\lambda_u \sqrt{\zeta}\!\left(\frac{\pi}{2} - \arctan\!\left(\frac{r^2}{\sqrt{\zeta}}\right)\right),
	\end{align}
	where $\zeta \triangleq sP_uC$ and the last step uses $\int \frac{d\tau}{\tau^2 + a^2} = \frac{1}{a}\arctan(\tau/a)$. \hfill$\blacksquare$

	\emph{Proof of Theorem~\ref{thm:cond_coverage}.}
	From \eqref{eq:coverage_event} and the approximation \eqref{eq:Gamma_approx}:
	\begin{align}
		p_{\mathrm{cov}}(\theta|r)
		&\approx \mathbb{E}_{I}\!\left[1 - \left(1 - e^{-\nu\,s(r,\theta)(I+\sigma_{\mathrm{eff}}^2)}\right)^{M_r}\right].
	\end{align}
	Expanding via the binomial theorem $1 - (1-e^{-y})^{M_r} = \sum_{k=1}^{M_r}\binom{M_r}{k}(-1)^{k+1}e^{-ky}$ and substituting $y = \nu\,s(r,\theta)(I + \sigma_{\mathrm{eff}}^2)$:
	\begin{align}
		p_{\mathrm{cov}}(\theta|r) &\approx \sum_{k=1}^{M_r}\binom{M_r}{k}(-1)^{k+1} e^{-k\nu\,s\,\sigma_{\mathrm{eff}}^2}\,\mathbb{E}\!\left[e^{-k\nu\,s\,I}\right] \nonumber\\
		&= \sum_{k=1}^{M_r}\binom{M_r}{k}(-1)^{k+1} e^{-k\nu\,s\,\sigma_{\mathrm{eff}}^2}\,\mathcal{L}_{I|R_0}(k\nu s\,|\,r),
	\end{align}
	where the factorization uses the fact that $\sigma_{\mathrm{eff}}^2(r)$ is deterministic conditioned on $r$. \hfill$\blacksquare$

	\section{Proof of Proposition~\ref{prop:pcov_coup}}\label{app:coupling_proof}
	
	The coverage event is $\{\tilde{H}_0 > s(r,\theta)(I + \sigma_n^2)\}$. Substituting the partial-fraction CCDF \eqref{eq:CCDF_Htilde}:
	\begin{align}
		p_{\mathrm{cov}}^{(\mathrm{coup})}(\theta|r) &= \mathbb{E}_I\!\left[\sum_{k=1}^{M_r}\omega_k\,e^{-s(I+\sigma_n^2)/\lambda_k}\right] \nonumber\\
		&= \sum_{k=1}^{M_r} \omega_k\,e^{-s\sigma_n^2/\lambda_k}\,\mathbb{E}_I\!\left[e^{-sI/\lambda_k}\right].
	\end{align}
	Identifying $\mathbb{E}[e^{-sI/\lambda_k}] = \mathcal{L}_{I|R_0}^{(\mathrm{coup})}(s/\lambda_k \mid r)$ yields \eqref{eq:pcov_coup}. \hfill$\blacksquare$

	\bibliographystyle{IEEEtran}
	\bibliography{IEEEabrv,reference}

\end{document}